\documentclass[final,leqno,onefignum,onetabnum]{siamltex1213}

\usepackage{amssymb}
\usepackage{epsfig}
\usepackage{amsmath}
\usepackage{amsfonts}
\usepackage{cases}

\newcommand{\ds}{\displaystyle}
\newcommand{\ts}{\textstyle}
\newcommand{\real}{\operatorname{Re}}

\definecolor{rouge}{rgb}{1,0,0}
\definecolor{bleu}{rgb}{0,0,1}
\definecolor{vert}{rgb}{0,0.5,0}
\definecolor{clair}{rgb}{0,0.5,1}

\definecolor{mygreen}{rgb}{0.15,0.7,0.15}
\definecolor{myred}{rgb}{0.9,0.05,0.05}
\definecolor{myblue}{rgb}{0.0352,0.4981,0.6509}

\usepackage[normalem]{ulem}

\title{Diffusive approximation of a time-fractional Burger's equation in nonlinear acoustics} 

\author{Bruno Lombard\thanks{LMA, CNRS, UPR 7051, Aix-Marseille Univ., Centrale Marseille, F-13453 Marseille Cedex 13, France.  (\email{lombard@lma.cnrs-mrs.fr}).}
\and Denis Matignon\thanks{ISAE-Supa\' ero, University of Toulouse, BP 54032, 31055 Toulouse Cedex 4, France. (\email{Denis.Matignon@isae.fr}).}}

\begin{document}
\maketitle
\slugger{siap}{xxxx}{xx}{x}{x--x}

\begin{abstract}
A fractional time derivative is introduced into Burger's equation to model losses of nonlinear waves. This term amounts to a time convolution product, which greatly penalizes the numerical modeling. A diffusive representation of the fractional derivative is adopted here, replacing this nonlocal operator by a continuum of memory variables that satisfy local-in-time ordinary differential equations. Then a quadrature formula yields a system of local partial differential equations, well-suited to numerical integration. The determination of the quadrature coefficients is crucial to ensure both the well-posedness of the system and the computational efficiency of the diffusive approximation. For this purpose, optimization with constraint is shown to be a very efficient strategy. Strang splitting is used to solve successively the hyperbolic part by a shock-capturing scheme, and the diffusive part exactly. Numerical experiments are proposed to assess the efficiency of the numerical modeling, and to illustrate the effect of the fractional attenuation on the wave propagation.
\end{abstract}

\begin{keywords}
fractional derivatives, diffusive representation, nonlinear acoustics, Burger's equation, Strang splitting, shock-capturing schemes
\end{keywords}

\begin{AMS}
26A33, 35L60, 74J30.
\end{AMS}

\pagestyle{myheadings}
\thispagestyle{plain}
\markboth{B. LOMBARD AND D. MATIGNON}{FRACTIONAL BURGER'S EQUATION}


\section{Introduction}\label{SecIntro}

We investigate Burger's equation with a fractional time derivative $D_t^\alpha$:
\begin{equation}
\frac{\textstyle \partial u}{\partial t}+\frac{\textstyle \partial}{\textstyle \partial x}\left(a\,u+b\,\frac{\textstyle u^2}{\textstyle 2}\right)=-\varepsilon\,D^\alpha_t u, \hspace{1cm} \varepsilon \geq 0,\quad 0<\alpha<1.
\label{TM}
\end{equation}
$D_t^\alpha$ is a convolution product in time with a singular kernel \cite{Matignon08}. The l.h.s. of (\ref{TM}) is a standard transport equation, with linear advection at constant speed $a$ and a nonlinear quadratic term with coefficient $b$. The r.h.s. of (\ref{TM}) models linear losses and memory effects along the propagation. Since $\alpha<1$, the hyperbolic nature of Burger's equation is preserved.

Fractional Burger's equation with a fractional Laplacian with respect to space in the r.h.s. - instead of a fractional time derivative - have been largely investigated by many authors. Such equations model anomalous dispersion or diffusion \cite{Sousa09}, or sedimentation of particles \cite{Chen08}. In this case, theoretical results of existence, uniqueness, regularity and asymptotic behavior of the solution can be found in \cite{Biler98,Guesmia10}. On the contrary, very few theoretical investigations of (\ref{TM}) have been proposed in the litterature, up to our knowledge. The particular case $\alpha=1/2$ has been examined in \cite{Sugimoto91}, where a matched-asymptotic analysis of the boundary layer is proposed, together with a semi-analytical resolution.

Various physical configurations are described by (\ref{TM}). Particular values of $\varepsilon$ and $\alpha$ enable to recover Chester's equation describing propagation of finite-amplitude sound waves in tubes \cite{Chester64}, up to ${\cal O}(\varepsilon^2$) terms. This equation is widely used to model brass instruments (trombones, trumpets): the transport terms describe the steepening of waves, yielding the typical "brassy" effect \cite{Hirschberg96}, and the fractional term models the viscothermal losses at the wall of the duct \cite{Bruneau89, Menguy00}. Moreover, a linear fractional wave equation known as the Lokshin model \cite{Lok78,Polack91} has been studied in e.g. \cite{Matignon-These94,Matignon94p,Helie06-M3AS}, and can be seen as the superposition of two one-way fractional transport equations of this type. Other applications of (\ref{TM}) concern viscoelasticity, propagation in elastic-walled tubes, or more generally wave propagation in media with memory and complex rheological properties \cite{Szabo94,Mainardi10,Nasholm11}. See \cite{Makarov97} for a review on the physical models involving nonlinear and thermoviscous phenomena.

The numerical resolution of (\ref{TM}) requires adequate tools for both the hyperbolic part and the fractional part. On the one hand, the computation of scalar nonlinear hyperbolic PDE such as Burger's equation is now a mature subject, with a wide number of available efficient approaches, e.g. shock-capturing schemes \cite{LeVeque02}. On the other hand, the computation of the fractional part is less standard. A naive discretization of this term requires to store the entire variable history, which could sometimes be used for fractional ODEs, but is out of reach in practical situations for fractional PDEs. Another approach is commonly used, based on the Gr\"unwald-Letnikov approximation of fractional derivatives \cite{Sousa09,Yuste05}. However, the stability analysis of this multistep method may be intricate  \cite{Lub86}: von Neumann analysis requires to bound the characteristic roots of the amplification matrix, which is a tedious task, especially when coupled with a nonlinear equation \cite{Yuste05,Yuste06}. 

Here, we follow an alternative time-domain approach based on a diffusive representation of the fractional derivative. The latter is written as a continuum of memory variables satisfying local-in-time ODE \cite{DeMi88,Staff94,Montseny98}. Discretization by a quadrature formula yields a diffusive approximation \cite{Deu10,Haddar10}, which is then coupled with the nonlinear hyperbolic equation. The stability of the system is obtained as long as the quadrature coefficients are positive. Positivity of the coefficients also ensures that the condition of numerical stability is the same as for the hyperbolic PDE, which constitutes a major advantage of this approach. 

The efficiency of the diffusive approximation relies crucially on the computation of the quadrature coefficients. The specifications concern both the positivity of the coefficients and the accuracy of the quadrature formula, in order to need only a small set of memory variables, and hence a reduced number of computational arrays. The methods based on Gaussian polynomials ensure positivity, but their convergence is very slow \cite{Yuan02,Diethelm08}, even if improvements have been recently obtained with Gauss-Jacobi polynomials \cite{Birk10}. Greater accuracy is reached when least-squares optimization is implemented \cite{Helie06-SP,Deu10,Blanc13a,Lombard14}, but some negative coefficients are usually obtained. In this paper, we use optimization with constraints of positivity, which provides a great improvement of accuracy compared with the aforementionned quadrature methods. This type of optimization has already been used with success in the context of poroelasticity \cite{TheseBlanc}, viscoelasticity \cite{Blanc16}, and recently for Chester's equation describing nonlinear acoustic waves in a guide \cite{Berjamin16}. 

Compared with previous works on nonlinear waves with fractional derivatives \cite{Lombard14,Berjamin16}, this paper introduces three novelties:
\begin{enumerate}
\item any value of $\alpha$ is considered, and not only $\alpha=1/2$;
\item contrary to Chester's equation, an energy functional is found, which ensures a solid theoretical basis;
\item in the linear case, a closed-form solution is proposed, which provides a strong validation of the numerical methods.
\end{enumerate} 
The paper is organized as follows. The model (\ref{TM}) is stated in section \ref{SecPb}. The diffusive representation of the fractional derivative is introduced. In section \ref{SecEvol}, the continuum of memory variables is discretized by a quadrature formula, yielding a local first-order system of PDEs. The positivity of the quadrature coefficients has a crucial influence on the properties of the system, such as the decrease of energy. The numerical methods are addressed in section \ref{SecNum}. The quadrature coefficients are initialized by a Gauss-Jacobi method, and then optimized under a positivity constraint. A splitting strategy is used to integrate the system of PDEs. The propagative part of the system is solved by a standard scheme for hyperbolic equations, whereas the diffusive part is solved exactly. Numerical experiments are proposed in section \ref{SecRes}. Comparisons with exact solutions in the linear case confirm the accuracy of the modeling. The effect of fractional dissipation on the emergence of shocks is also illustrated. Conclusions are drawn in section \ref{SecConclu}, and future lines of research are suggested. In appendix, the link between (\ref{TM}) and Chester's equation is shown and some properties are proven.


\section{Problem statement}\label{SecPb}

\subsection{Cauchy problem}\label{SecPbTM}

The problem at hand is
\begin{subnumcases}{\label{ToyModel}}
\ds
\frac{\textstyle \partial u}{\partial t}+\frac{\textstyle \partial}{\textstyle \partial x}\left(a\,u+b\,\frac{\textstyle u^2}{\textstyle 2}\right)+\varepsilon\,D^\alpha_t u=\delta(x)\,g(t),\quad t>0,\label{ToyModel1}\\
[10pt]
\ds
u(x,0)=u_0(x), \quad x\in \mathbb{R},\label{ToyModel2}
\end{subnumcases}
with $b \geq 0$ and $\varepsilon\geq 0$. The r.h.s. of (\ref{ToyModel1}) models a time forcing term located at $x=0$. For a causal function $h(t)$, $D^\alpha_t h$ refers to the Caputo fractional derivative in time of order $\alpha$, with $0<\alpha<1$:
\begin{equation}
D^\alpha_t h=\frac{\textstyle t^{-\alpha}}{\textstyle \Gamma(1-\alpha)}\mathop{*}\limits_{t}\frac{\textstyle dh}{\textstyle dt}=\frac{1}{\Gamma(1-\alpha)}\int_{0}^{t}(t-\tau)^{-\alpha}\frac{\textstyle dh}{\textstyle d\tau}(\tau)\,d\tau,
\label{FD}
\end{equation}
where $\Gamma$ is the Gamma Euler function and $\mathop{*}\limits_{t}$ is the convolution product in time. This definition follows from
\begin{equation}
D^\alpha_t h=I_t^{1-\alpha}\left(\frac{\textstyle dh}{\textstyle dt}\right),
\label{Caputo}
\end{equation}
where $I_t^\beta$ is the Riemann-Liouville fractional integral in time of order $\beta$, with $0<\beta<1$:
\begin{equation}
I^\beta_t h=\frac{1}{\Gamma(\beta)}\int_{0}^{t}(t-\tau)^{\beta-1}h(\tau)\,d\tau.
\label{FI}
\end{equation}


\subsection{Dispersion analysis}\label{SecPbDisp}

The goal of this section is to derive a dispersion analysis of the model (\ref{ToyModel}), that will serve as a reference case for further numerical approximations; hence functions to be transformed are supposed to be smooth enough, and the initial conditions accordingly. The forcing term is removed in this section: $g(t)=0$.

The Fourier transforms in time and space are denoted
\begin{equation}
\widehat{h}(\omega)={\cal F}_t(h)=\int_{-\infty}^{+\infty}h(t)\,e^{-i\omega t}\,dt,\hspace{0.8cm}
\widehat{h}(k)={\cal F}_x(h)=\int_{-\infty}^{+\infty}h(x)\,e^{+ikx}\,dx,
\label{Fourier}
\end{equation}
where $\omega$ is the angular frequency and $k$ is the wavenumber. The Fourier transform in time of the Caputo fractional derivative (\ref{FD}) is
\begin{equation}
\widehat{D_t^\alpha h}=(i\omega)^\alpha\hat{h}(\omega).
\label{FourierFD}
\end{equation}
Applying (\ref{Fourier}) to the fractional PDE (\ref{ToyModel}) and using (\ref{FourierFD}) provides the nonlinear equation 
\begin{equation}
i\,\omega\,\hat{u}-i\,k\left(a\,{\hat u}+\frac{\textstyle b}{\textstyle 2}\widehat{u^2}\right)+\varepsilon\,\chi(\omega)\,\hat{u}=0, 
\label{HatU}
\end{equation}
where $\chi$ is the symbol of the pseudo-differential operator (\ref{FD}): 
\begin{equation}
\chi(\omega)=\left(i\,\omega\right)^\alpha.
\label{ChiFD}
\end{equation}
When $a\neq 0$ and $b=0$, one obtains the dispersion relation
\begin{equation}
k=\frac{\textstyle \omega}{\textstyle a}-i\,\frac{\textstyle \varepsilon}{\textstyle a}\,\chi(\omega).
\label{Dispersion}
\end{equation}
It follows the phase velocity $\upsilon_\varphi=\omega\,/\,\Re\mbox{e}(k)$ and the attenuation $\eta=-\Im\mbox{m}(k)$:
\begin{equation}
\upsilon_\varphi(\omega)=\frac{\textstyle a}{\textstyle 1+\displaystyle \varepsilon \sin\left(\frac{\textstyle \alpha\,\pi}{\textstyle 2}\right)\,\omega^{\alpha-1}},\hspace{1cm}
\eta=\frac{\textstyle \varepsilon}{\textstyle a}\,\cos\left(\frac{\textstyle \alpha\,\pi}{\textstyle 2}\right)\,\omega^\alpha.
\label{VitAtt}
\end{equation}
One deduces the elementary properties: if $\varepsilon \neq 0$, then
\begin{equation}
\begin{array}{l}
\displaystyle
\upsilon_\varphi(0)=0,\quad \lim_{\omega\rightarrow +\infty}\upsilon_\varphi(\omega)=a,\quad \upsilon_\varphi^{'}>0,\\
[8pt]
\displaystyle
\eta(0)=0,\quad \lim_{\omega\rightarrow+\infty}\eta(\omega)=+\infty,\quad \eta^{'}>0.
\end{array}
\label{PropertyDisp}
\end{equation}
Figure \ref{FigDispExact} illustrates (\ref{VitAtt}) and the properties (\ref{PropertyDisp}), for various values of $\alpha$: 1/3, 1/2 and 0.7. These values are chosen because closed-form solutions of fractional advection are known when $\alpha=1/3$ and $\alpha=1/2$: see section \ref{SecExactParti}. The attenuation increases with $\alpha$, contrary to the phase velocity.

\begin{figure}[htbp]
\begin{center}
\begin{tabular}{cc}
phase velocity $\upsilon_\varphi$ & attenuation $\eta$ \\
\hspace{-0.8cm}
\includegraphics[scale=0.33]{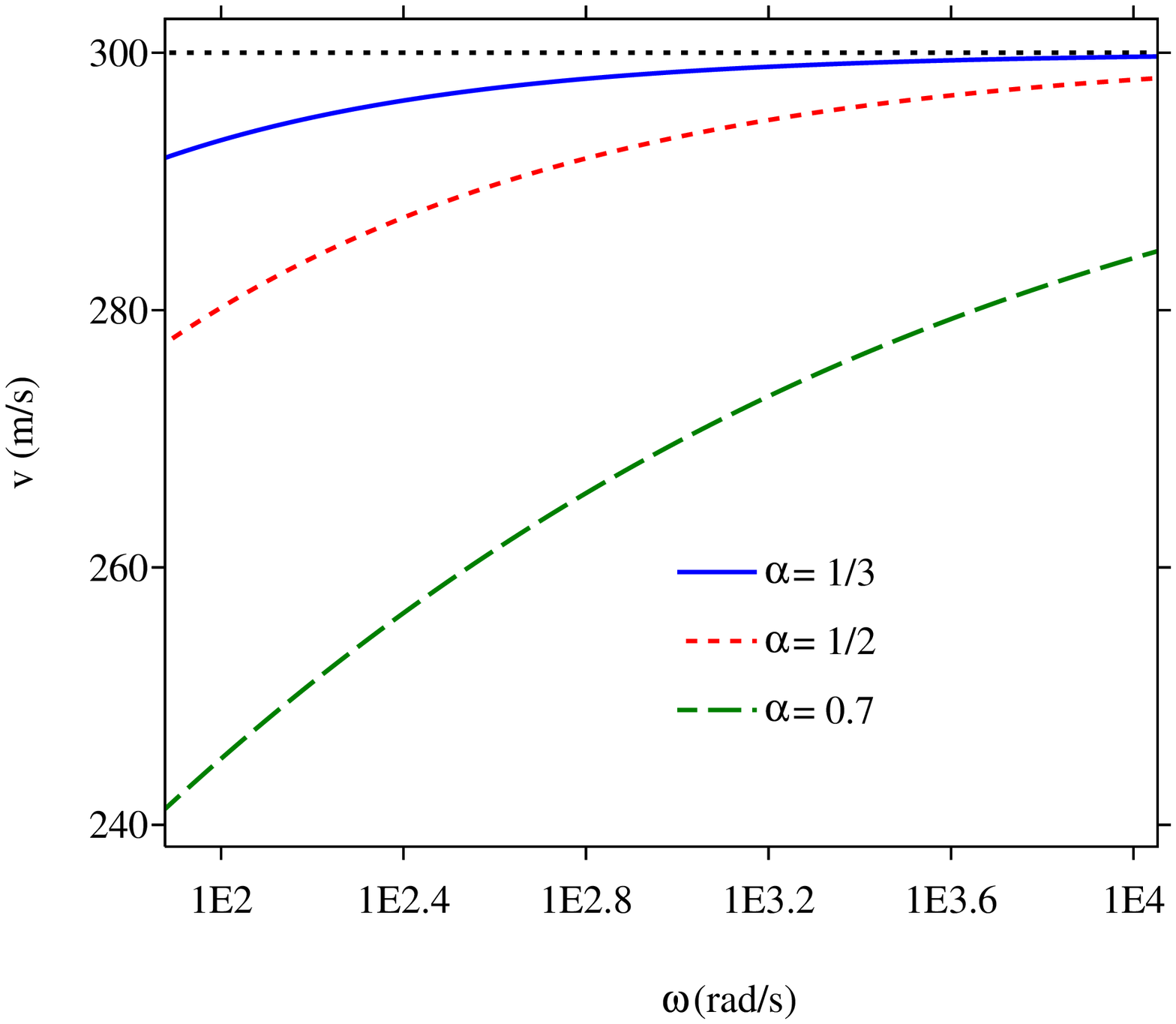}&
\hspace{-0.8cm}
\includegraphics[scale=0.33]{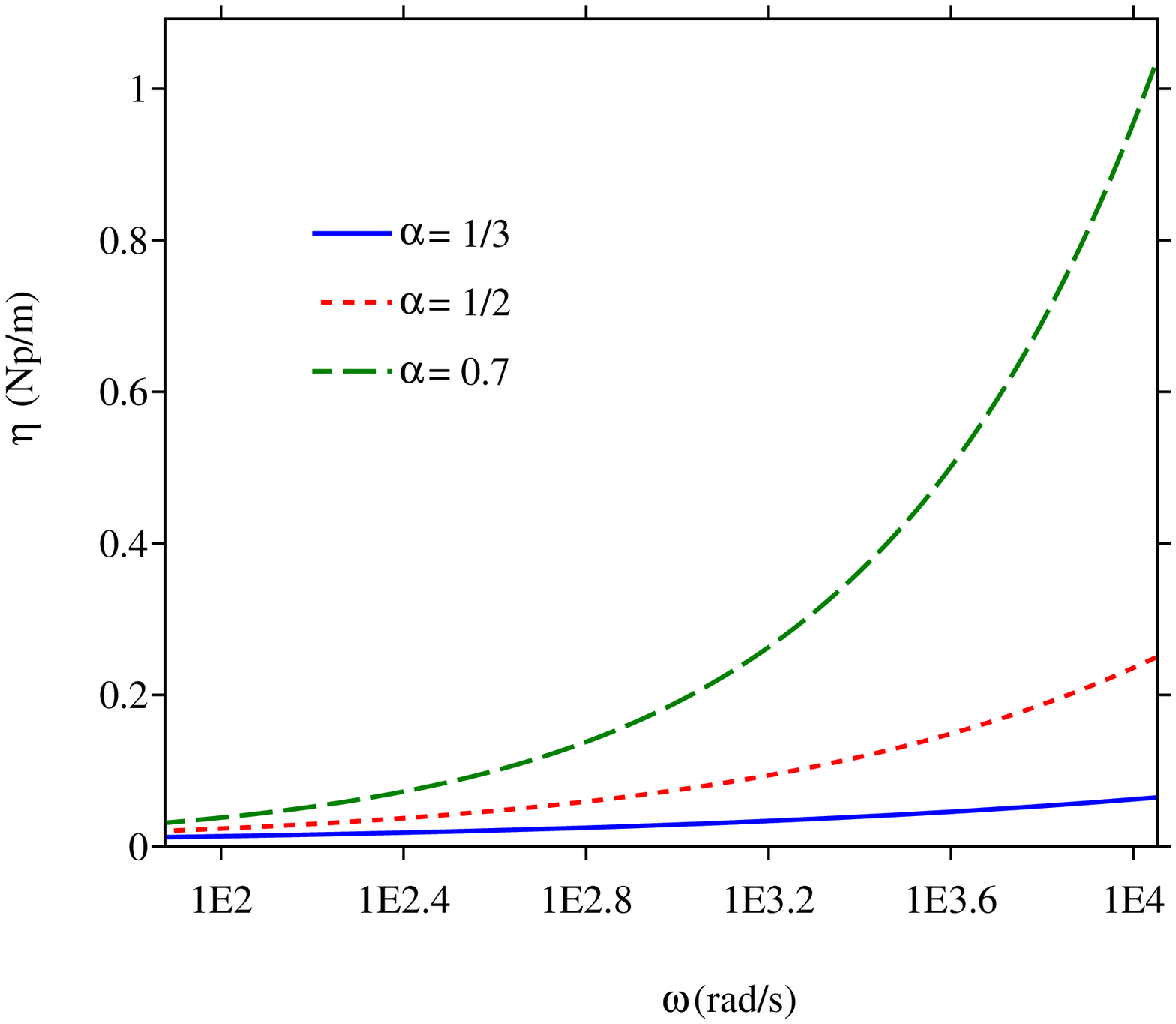}
\end{tabular}
\vspace{-0.5cm}
\caption{\label{FigDispExact} Dispersion curves deduced from (\ref{Dispersion}) in the linear regime ($b=0$), with $a=300$ m/s, $\varepsilon=1$ s$^{\alpha-1}$ and various values of $\alpha$: 1/3, 1/2 and 0.7. The horizontal dotted line in the phase velocity denotes the sound speed $a$.} 
\end{center}
\end{figure}


\subsection{Diffusive representation}\label{SecPbDR}

The convolution product in (\ref{FD}) complicates the numerical resolution of (\ref{ToyModel}). The past values of the solution must be stored, which is too expensive computationally. The alternative approach adopted in this study is based on a diffusive representation of fractional derivatives originally introduced in \cite{DeMi88,Staff94,Montseny98,Matignon08}. Here we follow the approach introduced in \cite{Yuan02,Diethelm08}, which proves to be equivalent to the diffusive representation formalism, up to the change of variables $\xi=\theta^2$: namely, for $0<\alpha<1$, the fractional derivative (\ref{FD}) can be recast as
\begin{equation}
D_t^\alpha u=\int_{0}^{+\infty}\!\!{\phi(x,t,\theta)\,d\theta},
\label{DR}
\end{equation}
where the function $\phi$ is defined owing to a change of variables as
\begin{equation}
\phi(x,t,\theta)=\frac{2\sin(\pi\alpha)}{\pi}\,\theta^{2\alpha-1}\int_{0}^t\frac{\textstyle \partial u}{\textstyle \partial \tau}(x,\tau)\,e^{-(t-\tau)\,\theta^2}\,d\tau.
\label{VarDiff}
\end{equation}
A short proof of (\ref{DR}) is given in appendix \ref{SecProofDR}. As $\phi$ is expressed in terms of an integral operator with decaying exponential kernel, it is referred to as a diffusive variable (or memory variable). From equation (\ref{VarDiff}), it satisfies the following first-order differential equation for $\theta>0$:
\begin{equation}
\left\{
\begin{array}{l}
\displaystyle
\frac{\textstyle \partial \phi}{\textstyle \partial t}=-\theta^2\,\phi+\gamma_\alpha\,\theta^{2\alpha-1}\frac{\textstyle \partial u}{\partial t},\\
[8pt]
\phi(x,0,\theta)=0,
\end{array}
\right.
\label{ODE-DR}
\end{equation}
with
\begin{equation}
\gamma_\alpha=\frac{2\sin(\pi\alpha)}{\pi} > 0.
\label{Gamma}
\end{equation}
Note for further use that $\gamma_{1-\alpha}=\gamma_{\alpha}$. The diffusive representation (\ref{DR})--(\ref{VarDiff}) amounts to replace the non-local term in (\ref{ToyModel}) by an integral over $\theta$ of the function $\phi(x,t,\theta)$ which obeys the local first-order ordinary differential equation (\ref{ODE-DR}).

For further analysis, one can also define another diffusive representation of the fractional derivative of order $\alpha$ (\ref{FD}), making use of the fractional integral (\ref{FI}) of order  $\beta=1-\alpha$: let $\psi$ be the new diffusive variable satisfying the ODE
\begin{equation}
\left\{
\begin{array}{l}
\ds
\frac{\textstyle \partial \psi}{\textstyle \partial t}=-\theta^2\,\psi+\gamma_\beta\,\theta^{1-2\beta}\,u,\\
[8pt]
\psi(x,0,\theta)=\Psi(x,\theta).
\end{array}
\right.
\label{ODE_Xi}
\end{equation}

\begin{proposition}
The following first identity holds:
\begin{equation}
\text{for} \quad \Psi(x,\theta):=0, \quad 
I_t^\beta u=\int_{0}^{+\infty}\psi(x,t,\theta)\,d\theta.
\label{DR_I}
\end{equation}
With a particular choice of initial data, the following second identity holds:
\begin{equation}
\text{for} \quad \Psi(x,\theta):= \gamma_\beta\frac{u_0(x)}{\theta^{1+2\beta}}, \quad 
D_t^\alpha u= \ds \int_{0}^{+\infty} \left(-\theta^2\,\psi(x,t,\theta)+\gamma_\beta\,\theta^{1-2\beta}\,u(x,t)\right)\,d\theta.
\label{DiffusifEtendu}
\end{equation}
\label{PropRD}
\end{proposition}
This latter representation  (\ref{ODE_Xi})-(\ref{DiffusifEtendu}) is an extended diffusive representation. The proof of proposition \ref{PropRD} is given in appendix \ref{SecProofRD}. \newpage


\section{Evolution equations}\label{SecEvol}

\subsection{Diffusive approximation}\label{SecEvolDA}

The integral in (\ref{DR}) is approximated by a quadrature formula on $L$ points, where the diffusive variables $\phi_j$ satisfy ODE deduced from (\ref{ODE-DR}):
\begin{equation}
\left\{
\begin{array}{l}
\displaystyle
D^{\alpha}_t u(x,t)\approx \sum_{\ell=1}^{L} \mu_\ell\,\phi(x,t,\theta_\ell)\equiv \sum_{\ell=1}^{L} \mu_\ell\,\phi_{\ell}(x,t),\\
[12pt]
\displaystyle
\frac{\textstyle \partial\phi_\ell}{\textstyle \partial t}=-\theta_\ell^2\,\phi_\ell+\gamma_\alpha\,\theta_\ell^{2\alpha-1}\,\frac{\textstyle \partial u}{\partial t},\qquad \ell=1,\cdots, L,\\
[10pt]
\displaystyle
\phi_\ell(x,0)=0.
\end{array}
\right.
\label{DA}
\end{equation}
Adequate choice of the weights $\mu_\ell$ and nodes $\theta_\ell$ is a crucial issue for the efficiency and accuracy of the diffusive approximation (\ref{DA}). It is discussed later in section \ref{SecNumQuad}.

Injecting the diffusive approximation (\ref{DA}) into (\ref{ToyModel}) yields the first-order system ($j=1,\cdots, L$)
\begin{subnumcases}{\label{EDP}}
\ds
\frac{\textstyle \partial u}{\partial t}+\frac{\textstyle \partial}{\textstyle \partial x}\left(a\,u+b\,\frac{\textstyle u^2}{\textstyle 2}\right)=-\varepsilon\,\sum_{\ell=1}^L\mu_\ell\,\phi_\ell+\delta(x)\,g(t),\label{EDP1}\\
[6pt]
\ds
\frac{\textstyle \partial \phi_j}{\partial t}+\gamma_{\alpha}\,\theta_j^{2\alpha-1}\,\frac{\textstyle \partial}{\textstyle \partial x}\left(a\,u+b\,\frac{\textstyle u^2}{\textstyle 2}\right)=-\theta_j^2\,\phi_j-\gamma_{\alpha}\,\theta_j^{2\alpha-1}\varepsilon\,\sum_{\ell=1}^L\mu_\ell\,\phi_\ell\label{EDP2}\\
[6pt]
\ds \hspace{5.1cm}\nonumber -\gamma_{\alpha}\,\theta_j^{2\alpha-1}\delta(x)\,g(t),\\
[6pt]
\ds
u(x,0)=u_0(x),\quad \phi_j(x,0)=0.\label{EDP3}
\end{subnumcases}
Taking the vectors of $(L+1)$ unknowns, forcing and initial data
\begin{equation}
\begin{array}{l}
\ds 
{\bf U}(x,t)=\left(u,\,\phi_1,\,\hdots,\,\phi_L\right)^T, \\
[8pt]
\ds
{\bf G}(t)=\left(g(t),\,-\gamma_\alpha\,\theta_1^{2\alpha-1}\,g(t),\hdots,\,-\gamma_\alpha\,\theta_L^{2\alpha-1}\,g(t)\right)^T,\\
[8pt]
\ds
{\bf U}_0(x)=\left(u_0(x),\,0,\,\hdots,\,0\right)^T,
\end{array}
\label{VecU}
\end{equation}
the system (\ref{EDP}) can be put in the form
\begin{equation}
\left\{
\begin{array}{l}
\ds
\frac{\textstyle \partial}{\textstyle \partial t}{\bf U}+\frac{\textstyle \partial}{\textstyle \partial x}{\bf F}({\bf U})={\bf S\,U}+\delta(x)\,{\bf G}(t),\\
[12pt]
\ds
{\bf U}(x,0)={\bf U}_0(x),
\end{array}
\right.
\label{SystHyper}
\end{equation}
where ${\bf F}=(F^{(1)},\,F^{(2)}, \hdots, F^{(L+1)})^T$ is the nonlinear flux function 
\begin{equation}
F^{(1)}=a\,u+b\,\frac{\textstyle u^2}{\textstyle 2}, \qquad
F^{(j)}=\gamma_{\alpha}\,\theta_{j-1}^{2\alpha-1}\,F^{(1)},\qquad j=2,\,\cdots,\, L+1.
\label{Fnonlin}
\end{equation}
${\bf S}$ is the $(L+1)\times(L+1)$ relaxation matrix
\begin{equation}
{\bf S}=-
\left(
\begin{array}{cccc}
0 & \varepsilon\,\mu_1 & \cdots & \varepsilon\,\mu_L\\
[8pt]
0 & \theta_1^2+\varepsilon\,\gamma_\alpha\,\theta_1^{2\alpha-1}\,\mu_1 & \cdots & \varepsilon\,\gamma_\alpha\,\theta_1^{2\alpha-1}\,\mu_L\\
[8pt]
\vdots & \vdots & \ddots & \vdots \\
[8pt]
0 & \varepsilon\,\gamma_\alpha\,\theta_L^{2\alpha-1}\,\mu_1 & \hdots & \theta_L^2+\varepsilon\,\gamma_\alpha\,\theta_L^{2\alpha-1}\,\mu_L
\end{array}
\right),
\label{MatS}
\end{equation}
containing the coefficients of the diffusive approximation (\ref{DA}). The size of ${\bf U}$ - and hence the number of computational arrays - increases linearly with the number of diffusive variables, which renders crucial the choice of a small value of $L$.


\subsection{Properties}\label{SecEvolProp}

Some elementary properties are stated about the evolution equations (\ref{EDP}) and the system (\ref{SystHyper}). First, applying Fourier transforms in time and space to (\ref{EDP}) provides the same dispersion relation than in (\ref{HatU}) or (\ref{Dispersion}). The only modification concerns $\chi$: instead of (\ref{ChiFD}), the symbol of the diffusive operator is
\begin{equation}
\tilde{\chi}(\omega)=\gamma_\alpha\,i\omega\sum_{\ell=1}^L\mu_{\ell}\frac{\textstyle \theta_{\ell}^{2\alpha-1}}{\textstyle \theta_{\ell}^2+i\,\omega}.
\label{ChiDA}
\end{equation}
Setting
\begin{equation}
K_{\alpha,L}=\gamma_\alpha\sum_{\ell=1}^L\mu_{\ell}\,\theta_{\ell}^{2\alpha-1},
\label{KalphaL}
\end{equation}
one has
\begin{equation}
\tilde{\chi}(\omega)\mathop{\sim}\limits_{0}K_{\alpha,L}\,i\omega,\hspace{0.8cm}\tilde{\chi}(\omega)\mathop{\sim}\limits_{+\infty}K_{\alpha,L}.
\label{BF-HF}
\end{equation}
These limit cases differ from the low-frequency and high-frequency behaviors of the exact symbol (\ref{ChiFD}).

Second, the hyperbolicity of the homogeneous system obtained with ${\bf S}={\bf 0}$ in (\ref{SystHyper}) is analysed.
\begin{proposition}
The system (\ref{SystHyper}) is hyperbolic but not strictly hyperbolic.
\label{PropHyp}
\end{proposition}

\begin{proof}
The eigenvalues $\zeta_\ell$ of the Jacobian matrix ${\bf J}=\frac{\partial {\bf F}}{\partial {\bf U}}$ are real:
\begin{equation}
\zeta_1=a+b\,u,\hspace{1cm} \zeta_\ell=0 \quad (\ell=2,\,\cdots,\, L+1).
\label{Eigenvalues}
\end{equation}
If $a+b\,u\neq 0$, then the matrice of eigenvectors ${\bf R}=({\bf r}_1|{\bf r}_2|\hdots|{\bf r}_{L+1})$ and its inverse ${\bf R}^{-1}$ are
\begin{equation}
{\bf R}=
\left(
\begin{array}{cccc}
1 & 0 & \hdots & 0\\
\gamma_\alpha\,\theta_1^{2\alpha-1} & 1 & & \\
\vdots & & \ddots & \\
\gamma_\alpha\,\theta_L^{2\alpha-1} &   & & 1
\end{array}
\right),\hspace{0.5cm}
{\bf R}^{-1}=
\left(
\begin{array}{cccc}
1 & 0 & \hdots & 0\\
-\gamma_\alpha\,\theta_1^{2\alpha-1} & 1 & & \\
\vdots & & \ddots & \\
-\gamma_\alpha\,\theta_L^{2\alpha-1} &   & & 1
\end{array}
\right).
\label{MatR}
\end{equation}
If $a+b\,u\neq 0$, then ${\bf R}={\bf R}^{-1}={\bf I}_{L+1}$,
where ${\bf I}$ is the identity matrix. 
\end{proof}\\

From (\ref{Eigenvalues}) and (\ref{MatR}), it follows that the characteristic fields satisfy:
\begin{equation}
\nabla \zeta_1=b,\hspace{1cm}
\nabla \zeta_\ell=0,\hspace{0.5cm}
\ell=1,\cdots,\,L.
\label{Characteristic}
\end{equation}
Consequently, there exists 1 genuinely nonlinear wave if $b \neq 0$ (shock wave or rarefaction wave), and $L$ linearly degenerate waves (contact discontinuities).\\

Third, we examine the energy of the system (\ref{EDP}) without forcing: $g(t)=0$. For this purpose, a quadrature formula of the extended diffusive representation (\ref{DiffusifEtendu}) is introduced, with $\psi_\ell(x,t)=\psi(x,t,\theta_\ell)$:
\begin{equation}
\left\{
\begin{array}{l}
\ds
D^{\alpha}_t u(x,t)\approx \sum_{\ell=1}^L\mu_\ell\frac{\partial \psi_\ell}{\partial t}=\sum_{\ell=1}^L \mu_\ell\left(-\theta_\ell^2\,\psi_\ell+\gamma_\beta\,\theta_\ell^{1-2\beta}\,u\right),\\
[12pt]
\ds
\frac{\textstyle \partial \psi_\ell}{\textstyle \partial t}=-\theta_\ell^2\,\psi_\ell+\gamma_\beta\,\theta_\ell^{1-2\beta}\,u,\qquad \ell=1,\cdots,L,\\
[12pt]
\ds
\psi_\ell(x,0)=\gamma_\beta\frac{u_0(x)}{\theta_\ell^{1+2\beta}}.
\end{array}
\right.
\label{ODE_extended}
\end{equation}

\begin{proposition}[Decrease of energy]
Let $u$ be a $C^1$ in space and time solution of (\ref{EDP}), and
\begin{equation}
\begin{array}{l}
\displaystyle
{\cal E}={\cal E}_1+{\cal E}_2,\\
[6pt]
\displaystyle
{\cal E}_1=\frac{\textstyle 1}{\textstyle 2}\int_\mathbb{R}u^2\,dx,\\
[8pt]
\displaystyle
{\cal E}_2=\frac{\textstyle 1}{\textstyle 2}\sum_{\ell=1}^L\int_\mathbb{R}\frac{\textstyle \varepsilon}{\textstyle \gamma_\alpha}\mu_\ell\,\theta_\ell^{3-2\alpha}\psi_\ell^2\,dx.
\end{array}
\label{NrjDA1}
\end{equation}
where the $\psi_\ell$ satisfy (\ref{ODE_extended}). Without forcing, one has
\begin{equation}
\frac{\textstyle d{\cal E}}{\textstyle dt}=-\sum_{\ell=1}^L\int_\mathbb{R}\frac{\textstyle \varepsilon}{\textstyle \gamma_\alpha}\,\mu_\ell\,\theta_\ell^{1-2\alpha}\left(\frac{\textstyle \partial \psi_\ell}{\textstyle \partial t}\right)^2\,dx.
\label{NrjDA2}
\end{equation}
\label{PropNrjDA}
\end{proposition}

\begin{proof}
One introduces the flux function $f$ and the Hamiltonian $H$
\begin{equation}
f(u)=au+b\frac{u^2}{2},\hspace{0.5cm}
H(u)=a\frac{u^2}{2}+b\frac{u^3}{3}.
\label{NrjHamilton}
\end{equation}
Equation (\ref{ToyModel}) and the extended diffusive approximation (\ref{ODE_extended}) yield the system
\begin{subnumcases}{\label{ProofNrjA}}
\ds \frac{\partial u}{\partial t}+\frac{\partial}{\partial x}f(u)+\varepsilon z=0,\label{ProofNrjA1}\\
\ds
z=\sum_{\ell=1}^L\mu_\ell\frac{\partial \psi_\ell}{\partial t},\label{ProofNrjA2}\\
\ds
\frac{\partial \psi_\ell}{\partial t}=-\theta_\ell^2\,\psi_\ell+\gamma_\beta\,\theta_\ell^{1-2\beta}\,u,\qquad \ell=1,\cdots,L,.\label{ProoNrjA3}
\end{subnumcases}
Taking the product of (\ref{ProofNrjA1}) with $u$ and integrating over $\mathbb{R}$ gives
\begin{equation}
\underbrace{\int_{\mathbb{R}}u\frac{\partial u}{\partial t}\,dx}_A+\underbrace{\int_{\mathbb{R}}u\,\frac{\partial}{\partial x}f(u)\,dx}_B+\underbrace{\int_{\mathbb{R}}\varepsilon\,u\,z\,dx}_C=0.
\label{ProofNrjB}
\end{equation}
The term $A$ in (\ref{ProofNrjB}) recovers the kinetic energy ${\cal E}_1$ in (\ref{NrjDA1}). For smooth solutions with compact support, the second term $B$ vanishes:
\begin{equation}
B=\left[H(u)\right]_{-\infty}^{+\infty}=0.
\label{ProofNrjC}
\end{equation}
Lastly, $u$ and $z$ in (\ref{ProofNrjB}) are expressed in terms of the extended diffusive variables (\ref{ProofNrjA2})-(\ref{ProoNrjA3}), and $\beta$ is replaced by $1-\alpha$. Since $\gamma_{1-\alpha}=\gamma_\alpha$, one obtains
\begin{equation}
\begin{array}{lll}
C &=& \ds \int_{\mathbb{R}}\frac{\varepsilon}{\gamma_\alpha}\theta_\ell^{1-2\alpha}\,\left(\frac{\textstyle \partial \psi_\ell}{\textstyle \partial t}+\theta_\ell^2\psi_\ell\right)\sum_{\ell=1}^L\mu_\ell\frac{\partial \psi_\ell}{\partial t}\,dx,\\
[12pt]
&=& \ds \sum_{\ell=1}^L \int_{\mathbb{R}}\frac{\varepsilon}{\gamma_\alpha}\left(\mu_\ell\,\theta_\ell^{1-2\alpha}\left(\frac{\partial \psi_\ell}{\partial t}\right)^2+ \theta_\ell^{3-2\alpha}\psi_\ell \frac{\partial \psi_\ell}{\partial t}\right)\,dx.
\end{array}
\end{equation}
It follows the term ${\cal E}_2$ in (\ref{NrjDA1}) and the decrease rate in (\ref{NrjDA2}), which concludes the proof.
\end{proof}

Two remarks are raised by proposition \ref{PropNrjDA}:
\begin{itemize}
\item the existence of a decreasing energy is conditional. Positivity of weights and nodes $\mu_\ell$ and $\theta_\ell$ is indeed required to ensure that ${\cal E}_2$ is a definite positive quadratic form (\ref{NrjDA1}) and to obtain $\frac{d{\cal E}}{dt}\leq 0$ (\ref{NrjDA2}). This positivity requirement is crucial for the well-posedness of (\ref{EDP}) and is examined in detail in section \ref{SecNumQuad}; 
\item $C^1$ smoothness of the solution both in space and time was assumed. In the case where a shock occurs, (for instance if $\varepsilon=0$), then the term $B$ in (\ref{ProofNrjC}) no more vanishes. It is replaced by $[u]^3/12<0$, where $[u]<0$ refers to the jump $u(x_s^+,t)-u(x_s^-,t)$, and $x_s(t)$ is the location of the shock. The decrease of energy is then the sum of two terms: a term proportional to $\varepsilon$ (due to intrinsic attenuation), and a term due to the occurence of shocks. To the best of our knowledge, there is still no theoretical results to predict the existence of shocks in the case $\varepsilon \neq 0$. Numerical experiments performed in section \ref{SecResShock} are a preliminary exploration of this property.\\
\end{itemize}

The fourth and last property concerns the eigenvalues of the relaxation matrix ${\bf S}$ in (\ref{MatS}).
\begin{proposition}
Let us assume that the nodes $\theta_\ell$ in (\ref{DA}) are sorted in increasing order :
$$
0<\theta_1<\theta_2<\cdots<\theta_L,
$$
and that the weights are positive: $\mu_\ell>0$. Then 0 is a simple eigenvalue of ${\bf S}$. Moreover, the $L$ nonzero eigenvalues $\lambda_\ell$ of ${\bf S}$ are real negative and satisfy:
\begin{equation}
\lambda_L<-\theta_L^2<\cdots<-\theta_{\ell+1}^2<\lambda_\ell<-\theta_\ell^2<\cdots<\lambda_1<-\theta_1^2<0.
\label{Encadre}
\end{equation}
\label{PropVpS}
\end{proposition}

The proof is given in appendix \ref{SecProofS}. Three remarks are raised by proposition \ref{PropVpS}:
\begin{enumerate}
\item for $L=1$, the eigenvalue is explicitly known:
\begin{equation}
\lambda_1=-\theta_1^2-\varepsilon\,\gamma_\alpha\,\theta_1^{2\alpha-1}\,\mu_1.
\label{N1}
\end{equation}
\item a lower bound of the spectral radius of ${\bf S}$ is obtained:
\begin{equation}
\varrho({\bf S})>\theta_L^2;
\label{RayonSpectral}
\end{equation}
\item as in proposition \ref{PropNrjDA}, the positivity of the weights $\mu_\ell$ is a crucial hypothesis. In the contrary case, one observes numerically that the eigenvalues of ${\bf S}$ do not satisfy (\ref{Encadre}). Moreover, complex conjugate roots can be obtained.
\end{enumerate}


\section{Numerical modeling}\label{SecNum}

\subsection{Numerical scheme}\label{SecNumSchem}

In order to integrate the system (\ref{SystHyper}), one introduces a uniform mesh size $\Delta x$ and a variable time step $\Delta t_n$. The approximation of the exact solution ${\bf U}(x_j=j\,\Delta x,t_n=t_{n-1}+\Delta t_n)$ is denoted by ${\bf U}_j^n$. Unsplit integration of (\ref{SystHyper}) is not optimal, because the stability condition typically implies \cite{LeVeque02}
\begin{equation}
\Delta t_n\leq \min\left(\frac{\textstyle \Delta x}{\textstyle a^n_{\max}},\,\frac{\textstyle 2}{\textstyle \varrho({\bf S})}\right),
\label{CFL}
\end{equation}
where $a^n_{\max}=a+b\,\max(u_j^n)$ is the maximum numerical velocity at time $t_n$. As shown in proposition \ref{PropVpS}, the spectral radius of the relaxation matrix $\varrho({\bf S})$ grows with the maximal node of quadrature (\ref{RayonSpectral}), penalizing the standard CFL condition. Moreover, solving directly (\ref{SystHyper})  requires to build an adequate scheme for the full system with source term.

\paragraph{Splitting}A more efficient strategy is adopted here. Equation (\ref{SystHyper}) is split into a hyperbolic step
\begin{equation}
\frac{\textstyle \partial}{\textstyle \partial t}{\bf U}+\frac{\textstyle \partial}{\textstyle \partial x}{\bf F}({\bf U})={\bf 0},
\label{SplitPropa}
\end{equation}
and a relaxation step
\begin{equation}
\frac{\partial}{\partial t}{\bf U}={\bf S}\,{\bf U}+\bf \delta(x)\,{\bf G}(t).
\label{SplitDiffu}
\end{equation}
The discrete operators to solve (\ref{SplitPropa}) and (\ref{SplitDiffu}) are denoted by ${\bf H}_a$ and ${\bf H}_b$, respectively. The Strang splitting \cite{LeVeque02,Holden11} is then used between $t_n$ and $t_{n+1}$, solving successively (\ref{SplitPropa}) and (\ref{SplitDiffu}) with adequate time increments:
\begin{equation}
\begin{array}{lllll}
\displaystyle
&\bullet& {\bf U}_{j}^{(1)}&=&{\bf H}_{b}\left(\frac{\Delta t_n}{ \ts2}\right)\,{\bf U}_{j}^{n},\\
[6pt]
\displaystyle
&\bullet& {\bf U}_{j}^{(2)}&=&{\bf H}_{a}\left(\Delta t_n\right)\hspace{0.2cm}{\bf U}_{j}^{(1)},\\
[6pt]
\displaystyle
&\bullet& {\bf U}_{j}^{n+1}&=&{\bf H}_{b}\left(\frac{\Delta t_n}{\ts 2}\right)\,{\bf U}_{j}^{(2)}.
\end{array}
\label{AlgoSplitting}
\end{equation}
Provided that ${\bf H}_a$ and ${\bf H}_b$ are second-order accurate and stable operators, the time-marching (\ref{AlgoSplitting}) gives a second-order accurate approximation of the original equation (\ref{SystHyper}).

\paragraph{Hyperbolic step}The homogeneous equation (\ref{SplitPropa}) is solved by a conservative scheme for nonlinear hyperbolic PDE:
\begin{equation}
\begin{array}{l}
\displaystyle
u_j^{n+1}=u_j^n-\frac{\textstyle \Delta t_n}{\textstyle \Delta x}\left(F^{(1)}_{j+1/2}-F^{(1)}_{j-1/2}\right),\\
[12pt]
\displaystyle
\phi_{j,\ell}^{n+1}=\phi_{j,\ell}^n-\gamma_\alpha\,\theta_{\ell}^{2\alpha-1}\frac{\textstyle \Delta t_n}{\textstyle \Delta x}\left(F^{(1)}_{j+1/2}-F^{(1)}_{j-1/2}\right), \hspace{1cm} \ell=1,\cdots,L,
\end{array}
\label{TVD}
\end{equation}
where $F^{(1)}_{j\pm1/2}$ is the numerical flux function of the advection-Burger's part in (\ref{EDP1}). In practice, a second-order TVD scheme with MC-limiter is used in our numerical experiments \cite{LeVeque02}. The stability analysis of (\ref{TVD}) yields the optimal CFL condition 
\begin{equation}
\Upsilon=\frac{\textstyle a^n_{\max}\,\Delta t_n}{\textstyle \Delta x} \leq 1.
\label{CFLopti}
\end{equation}

\paragraph{Relaxation step} Since $\varepsilon$ and the quadrature coefficients $\mu_\ell$, $\theta_\ell$ do not vary with time, ${\bf S}$ is constant in time and the relaxation step (\ref{SplitDiffu}) can be solved exactly. Without forcing, one obtains
\begin{equation}
{\bf H}_b\left(\frac{\Delta\,t}{2}\right)\,{\bf U}_j = e^{{\bf S}\frac{\Delta\,t}{2}}\,{\bf U}_j.
\label{SplitDiffuExp}
\end{equation}
The matrix exponential is computed numerically using a $(6,6)$ Pad\'e approximation in the {\it scaling and squaring method} \cite{Moler03}. Since $\mu_\ell>0$, proposition \ref{PropVpS} ensures that the eigenvalues of ${\bf S}$ are real negative; as a consequence, this approximation is stable. If the physical parameters are constant in space, as considered in the forthcoming numerical experiments, then ${\bf S}$ is constant. Therefore the computation (\ref{SplitDiffuExp}) needs to be done only once at each time step, leading to a negligible computational cost. This part of the splitting is unconditionally stable, 

\paragraph{Properties of the coupling}

The operators ${\bf H}_a$ and ${\bf H}_b$ are second-order accurate and exact, respectively. As a consequence, the Strang splitting (\ref{AlgoSplitting}) is second-order accurate.

The global stability requirement is (\ref{CFLopti}) and is not penalized by the relaxation step. In other words, the time step only depends on the advection and Burger's coefficients in (\ref{EDP}). In particular, $\Delta t_n$ does not depend on the coefficients of the diffusive representation. In practice, $\Delta t_n$is computed after  the second iteration of ${\bf H}_b$ by (\ref{CFLopti}).


\subsection{Quadrature coefficients}\label{SecNumQuad}

It remains to compute the set $\{(\mu_\ell,\theta_\ell)\}$ of $2L$ coefficients involved in the hyperbolic step (\ref{TVD}) and the relaxation step (\ref{SplitDiffuExp}). For this purpose, two different approaches can be employed. The most usual one is based on orthogonal polynomials, while the second approach is associated with an optimization process. Both lead to positive quadrature coefficients, which ensures the stability of (\ref{EDP}), as shown by propositions \ref{PropNrjDA} and \ref{PropVpS}. Here we will combine these two approaches: Gaussian formulae yield initial values of the coefficients, and then optimization with constraint is applied.

\paragraph{Gaussian quadrature} Various orthogonal polynomials can be used to evaluate the improper integral (\ref{DR}) introduced by the diffusive representation of fractional derivatives. Historically, the first one has been proposed in \cite{Yuan02}, where a Gauss-Laguerre quadrature is chosen. Its slow convergence was highlighted and then corrected in \cite{Diethelm08} with a Gauss-Jacobi quadrature. This latter method has been modified in \cite{Birk10}, where alternative weight functions are introduced, yielding an improved discretization of the diffusive variable owing to the use of an extended interpolation range. Following this latter modified Gauss-Jacobi approach, while omitting the time and space coordinates for the sake of brevity, the improper integral (\ref{DR}) is then recast as
\begin{equation}
\int_0^{+\infty}\phi(\theta)\,d\theta=\int_{-1}^{+1}\left(1-\tilde{\theta}\right)^\beta\left(1+\tilde{\theta}\right)^\delta\,\tilde{\phi}(\tilde{\theta})\,d\tilde{\theta}\simeq \sum_{\ell=1}^L\tilde{\mu}_\ell\,\tilde{\phi}(\tilde{\theta}_\ell),
\label{GJ1}
\end{equation}
with the modified diffusive variable $\tilde{\phi}$ defined as
$$
\tilde{\phi}(\tilde{\theta})=\frac{\textstyle 4}{\textstyle \left(1-\tilde{\theta}\right)^{\beta-1}\left(1+\tilde{\theta}\right)^{\delta+3}}\,\phi\left(\left(\frac{\textstyle 1-\tilde{\theta}}{\textstyle 1+\tilde{\theta}}\right)^2\right),
$$
and where the weights and nodes $\{(\tilde{\mu}_{\ell},\tilde{\theta}_{\ell})\}$ are computed by standard routines \cite{NRPAS}. According to the analysis of \cite{Birk10}, Section 4, an optimal choice for the coefficients in (\ref{GJ1}) is: $\beta=2\,\overline{\alpha}+1$ and $\delta=-(2\,\overline{\alpha}-1)$, with $\overline{\alpha}=2\,\alpha-1$. Equating the series (\ref{GJ1}) and (\ref{DA}) that both approximate the term \eqref{DR}, the quadrature coefficients are deduced:
\begin{equation}
\mu_\ell=\frac{\textstyle 4\,\tilde{\mu}_\ell}{\textstyle \left(1-\tilde{\theta}_\ell\right)^{\beta-1}\left(1+\tilde{\theta}_\ell\right)^{\delta+3}},\hspace{1cm}
\theta_\ell=\left(\frac{\textstyle 1-\tilde{\theta}_\ell}{\textstyle 1+\tilde{\theta}_\ell}\right)^2.
\label{Birk2}
\end{equation}


\paragraph{Optimization quadrature} As said in section \ref{SecEvolProp}, the dispersion relation of the fractional PDE (\ref{ToyModel}) and of its diffusive approximation (\ref{EDP}) differ only in the symbols (\ref{ChiFD}) and (\ref{ChiDA}) of the pseudo-differential operators. Equating these quantities provides a means to estimate the quadrature coefficients. It is recalled that the low-frequency and high-frequency limits of $\chi$ and $\tilde{\chi}$ differ; see (\ref{BF-HF}). Consequently, the optimisation procedure proposed here is valid only on a limited frequency range.

For a given number $K$ of angular frequencies $\omega_k$, one defines the following objective function
\begin{equation}
\begin{array}{lll}
\ds
{\cal J}_{L,K}\left(\{\mu_\ell,\theta_\ell)\}\right)&=&
\ds\sum_{k=1}^K\left|\frac{\textstyle \tilde{\chi}(\omega_k)}{\textstyle \chi(\omega_k)}-1\right|^2,\\
[8pt]
&=& \ds \sum_{k=1}^K\left|\gamma_\alpha\sum_{\ell=1}^L\mu_\ell\,\theta_\ell^{2\alpha-1}\,\frac{\textstyle (i\omega_k)^{1-\alpha}}{\textstyle \theta_\ell^2+i\omega_k}-1\right|^2,
\end{array}
\label{Objective}
\end{equation}
to be minimized w.r.t parameters $(\mu_\ell,\theta_\ell)$ for $\ell=1,\dots,L$. A straightforward linear minimization of \eqref{Objective} may lead to some negative parameters \cite{TheseBlanc,Blanc13a}, so that a nonlinear optimization with the positivity constraints $\mu_{\ell}\geq 0$ and $\theta_{\ell}\geq 0$ is preferred. 

An additional constraint is induced by the exponential of the matrix ${\bf S}$ in (\ref{SplitDiffuExp}). As noticed in (\ref{RayonSpectral}), large values of $\theta_\ell$ yields a large spectral radius of ${\bf S}$. In this case, the "scaling and squaring method" used to compute the exponential (\ref{SplitDiffuExp}) may be unstable. An additional constraint $\theta_{\ell}\leq \theta_{\text{max}}$ is therefore introduced to avoid the algorithm to diverge. 

The problem of minimization is nonlinear and non-quadratic w.r.t. abscissae $\theta_{\ell}$. To solve it, we use the algorithm SolvOpt \cite{Kappel00} based on the iterative Shor's method \cite{Shor85}. This method can be applied to a large class of functions, and in particular to (\ref{Objective}). It has been validated and applied to various applications, see e.g. \cite{Rekik11} and references therein. 

As for any local algorithm, Shor's method must be initialized with care. The initial values $\mu^{\,0}_{\ell}$ and $\theta^{\,0}_{\ell}$ are obtained by the modified Jacobi method \eqref{Birk2} for $\ell=1,\dots,L$. Doing so, the required positivity constraints are satisfied by the initial guesses, which are admissible solutions to (\ref{Objective}). 

Finally, the angular frequencies $\omega_k$ for $k = 1,\dots,K$ in (\ref{Objective}) are chosen linearly on a logarithmic scale over a given optimization band $[\omega_{\text{min}},\omega_{\text{max}}]$, i.e.
\begin{equation}
\omega_k = \omega_{\text{min}}\left( \frac{\omega_{\text{max}}}{\omega_{\text{min}}}\right)^{\!\frac{k-1}{K-1}}.
\label{OmegaK}
\end{equation}
In forthcoming numerical experiments, we use $\omega_{\text{min}}=\omega_c/10$ and $\omega_{\text{max}}=10\times\omega_c$. The parameter $\theta_{\text{max}}$ is set to $\theta_{\text{max}}=100\,\omega_{\text{max}}$. The number of angular frequencies is chosen equal to $K=2L$.

There is no theoretical argument justifying the choice of the interval $[\omega_{\min}=\omega_c/10, \omega_{\max}=10 \times \omega_c]$. It is only a reasonable choice, which can be sharpened depending on the application at hand. For instance, let us consider the simulation of resonators in musical acoustics \cite{Berjamin16}: then, the optimization range must be included in the range of interest lies in the audible spectrum [20 Hz, 20 kHz]. 

Concerning the upper limit of optimization, high frequencies are generated when $\varepsilon$ is small. Since the spectrum of the signal evolves, it seems strange at first glance to define a given upper limit of optimization. However, this problem exists already in the choice of the spatial discretization, even in the inviscid Burger's equation. Indeed, choosing the spatial mesh $\Delta x$ relies implicitly on the choice of a maximal sampling frequency. For higher frequencies, the number of grid nodes per wavelength is too small to give a reasonable approximation of the PDE under study, and the user assumes that this part of the signal is not useful. In other words, the choice of the upper range of optimization must be consistant with the choice of the spatial discretization.


\paragraph{Validation of the quadrature method}

\begin{figure}[htbp]
\begin{center}
\begin{tabular}{cc}
phase velocity $\upsilon_\varphi$ & attenuation $\eta$ \\
\hspace{-0.8cm}
\includegraphics[scale=0.33]{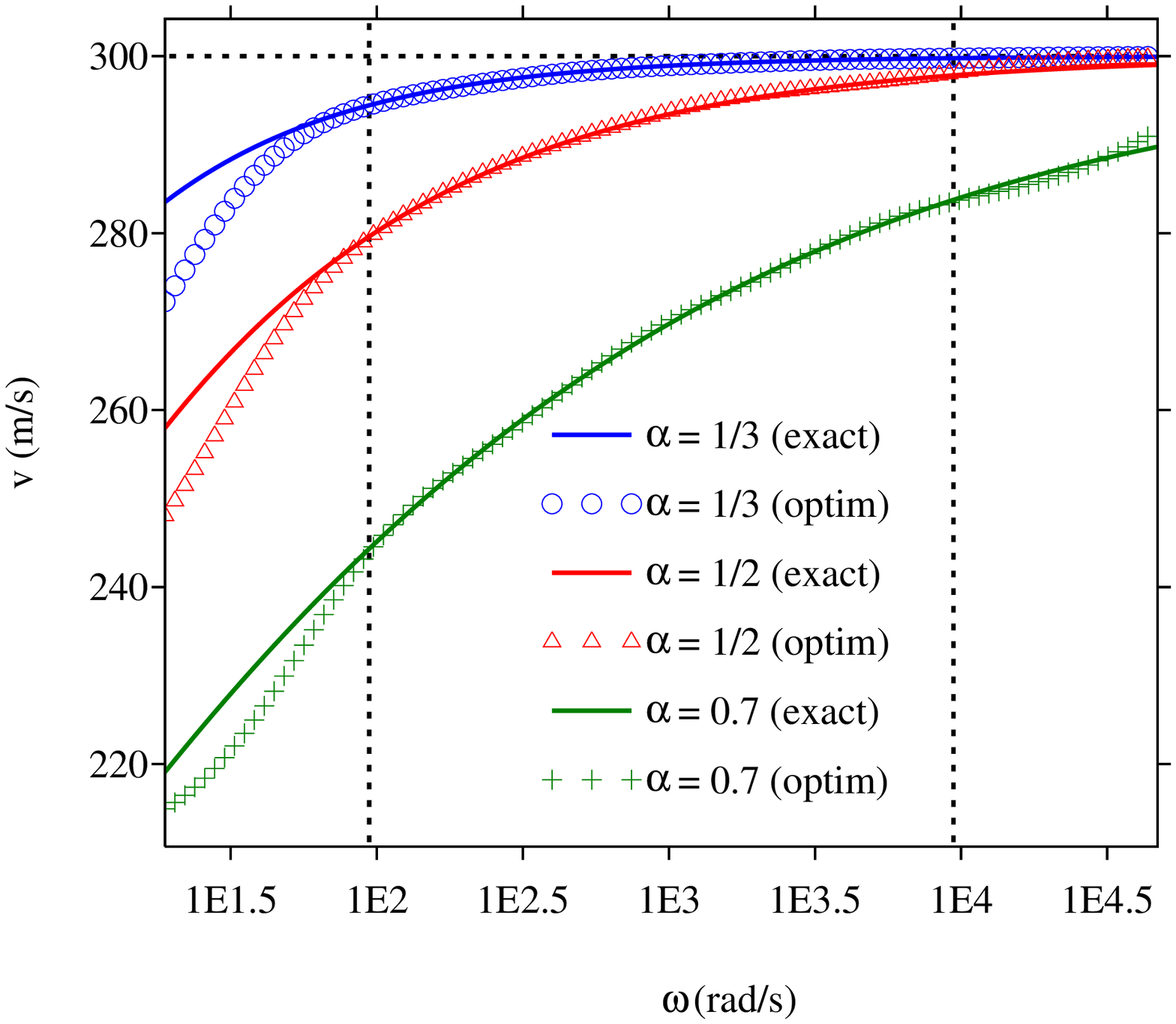}&
\hspace{-0.8cm}
\includegraphics[scale=0.33]{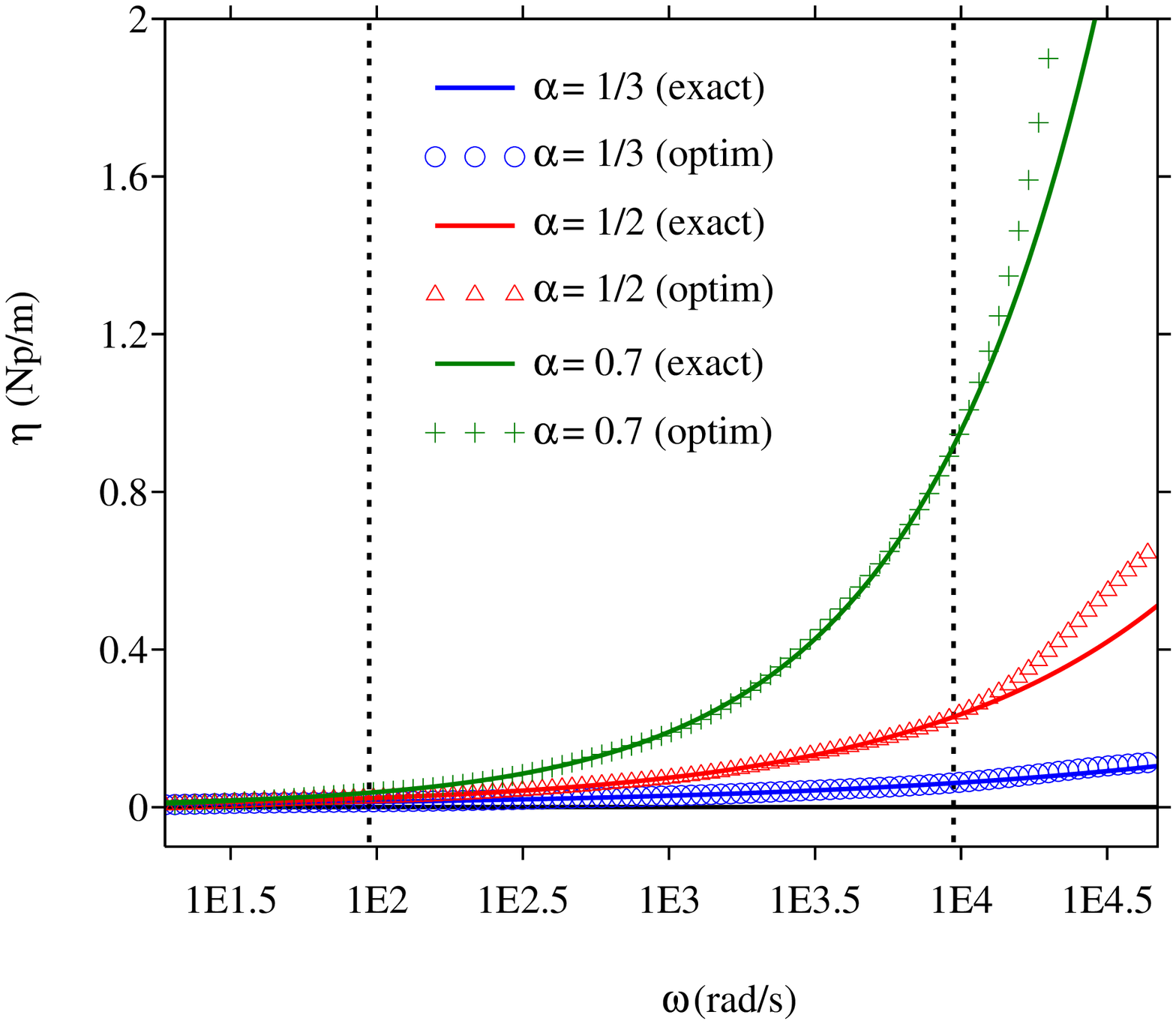}
\end{tabular}
\vspace{-0.5cm}
\caption{\label{FigDispOpti} Comparison between the dispersion curves of the fractional model (\ref{ToyModel}) and of the diffusive model (\ref{EDP}). The parameters are $a=300$ m/s, $b=0$, $\varepsilon=1$ s$^{\alpha-1}$, $\alpha=1/3$, 1/2 and 0.7, and $L=4$ diffusive variables. The quadrature coefficients are obtained with the optimization procedure. The horizontal dotted line denotes the sound velocity $a$. The vertical dotted lines denote the range of optimization.} 
\end{center}
\end{figure}

One considers a wave with a central frequency $f_c=150$ Hz, yielding $\omega_c=942.47$ rad/s (\ref{Mach}). The physical parameters are $a=300$ m/s, $b=0$, $\varepsilon=1$; various values of the fractional order $\alpha$ are investigated (1/3, 1/2 and 0.7). Figure \ref{FigDispOpti} compares the dispersion curves (\ref{Dispersion}) obtained with the exact symbol (\ref{ChiFD}) and the diffusive symbol (\ref{ChiDA}), respectively. Optimization with constraint of positivity is implemented. The results are displayed on the range $[\omega_{\min}/5,\omega_{\max}\times 5]$. Excellent agreement is obtained on $[\omega_{\min},\omega_{\max}]$, whatever the value of $\alpha$. The accuracy decreases outside the range of optimization. It follows from i) the optimization process, ii) the different low-frequency and high-frequency behaviors of the exact and diffusive symbols (see (\ref{BF-HF})). 

\begin{figure}[htbp]
\begin{center}
\begin{tabular}{cc}
(i) & (ii)\\
\hspace{-0.8cm}
\includegraphics[scale=0.33]{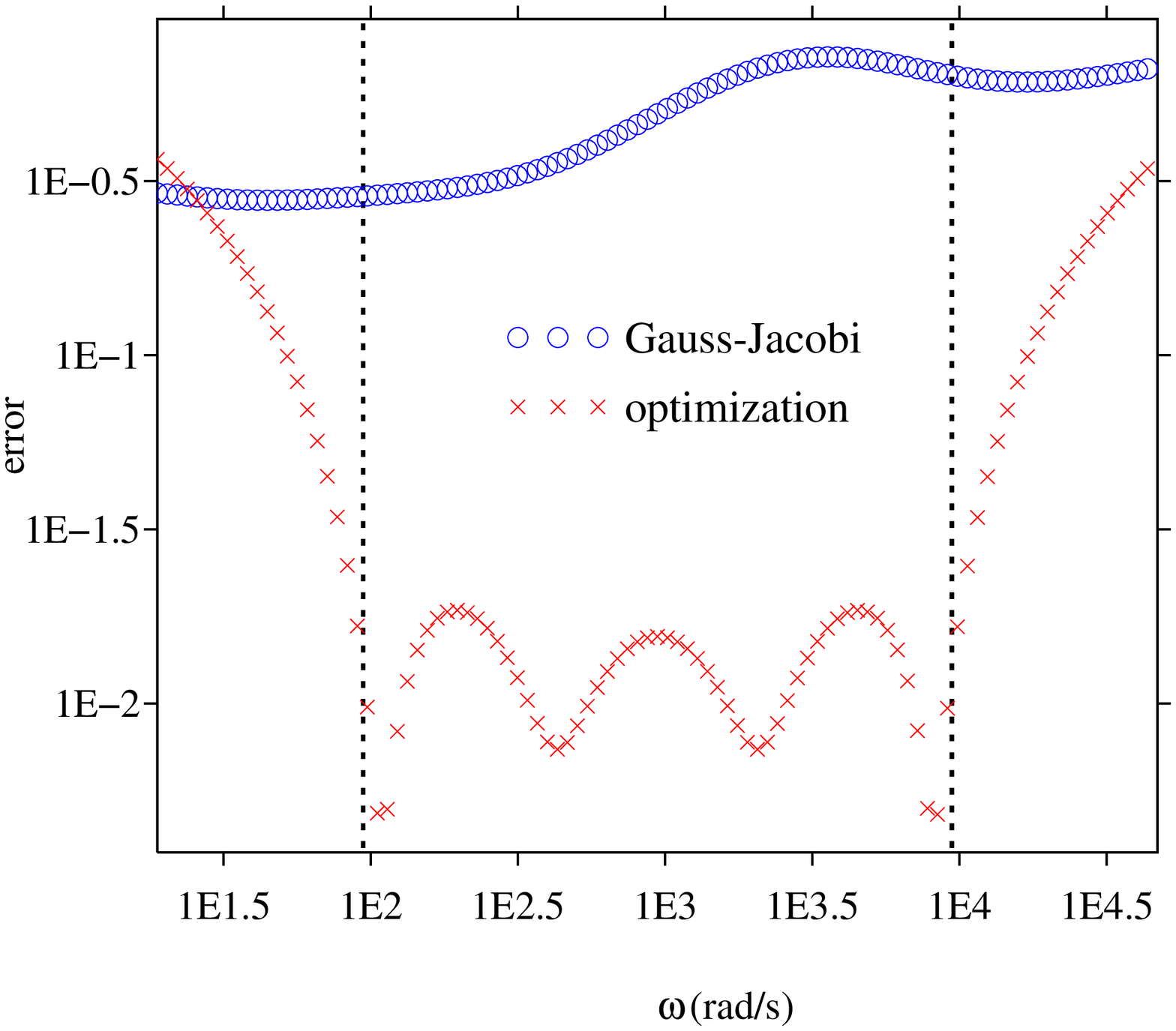}&
\hspace{-0.8cm}
\includegraphics[scale=0.33]{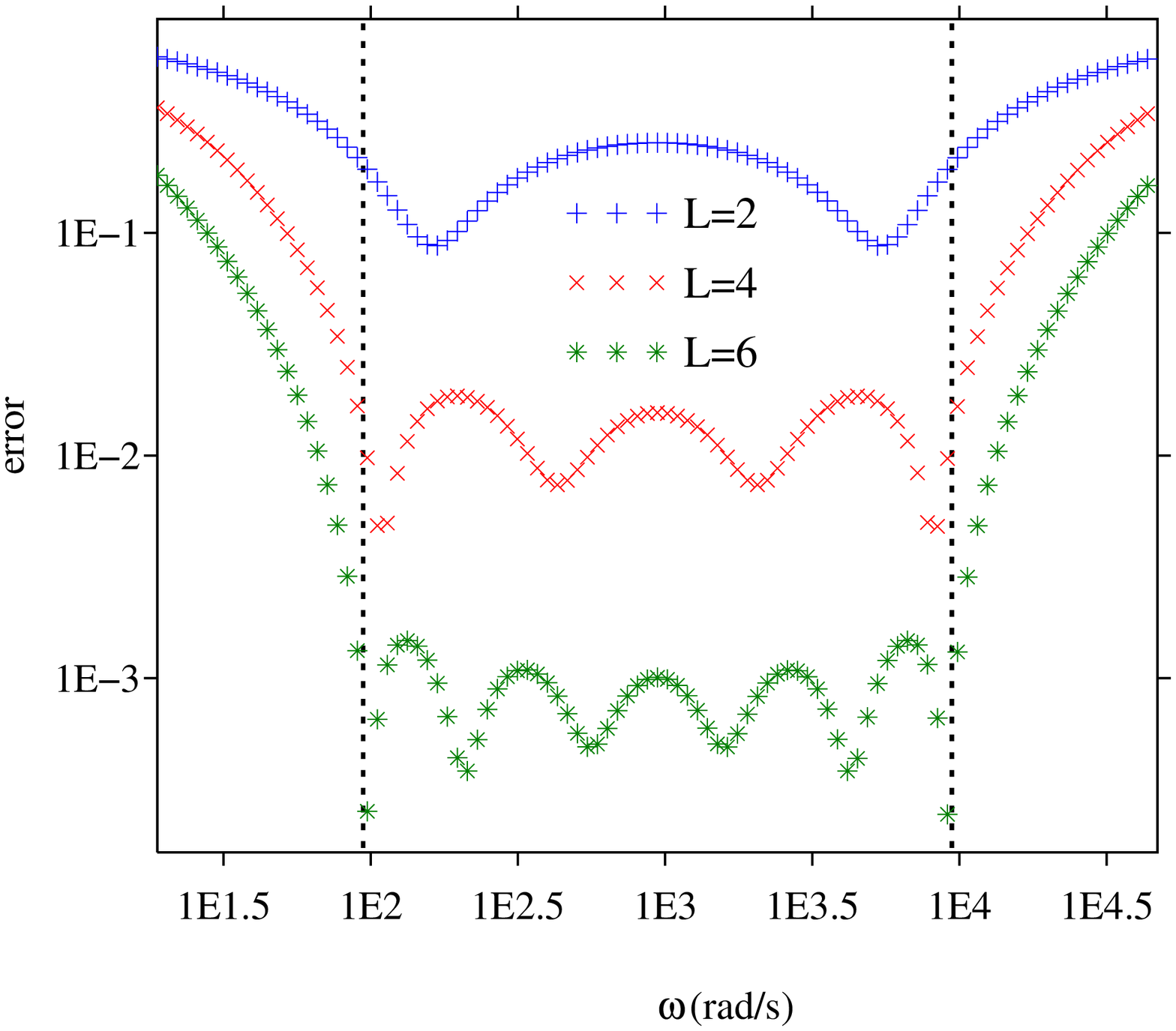}
\end{tabular}
\vspace{-0.5cm}
\caption{\label{FigErreur} Error of model $\left|\frac{\tilde{\chi}(\omega)}{\chi(\omega)}-1\right|$ deduced from (\ref{ChiFD}) and (\ref{ChiDA}). The parameters are $a=300$ m/s, $\varepsilon=1\,\mbox{s}^{-1/2}$ and $\alpha=0.5$. Left row (i): Gauss-Jacobi and optimization methods are compared, for $L=4$ diffusive variables. Right row (ii): optimization is used, and various values $L$ are considered. The vertical dotted lines denote the range of optimization.}
\end{center}
\end{figure}

The objective function (\ref{Objective}) is built by minimizing the error of model $\left|\frac{\tilde{\chi}(\omega)}{\chi(\omega)}-1\right|$ at discrete angular frequencies $\omega_k$. Figure \ref{FigErreur} illustrates this error for continuous values of $\omega$ and in the case $\alpha=0.5$. In (i), the influence of the quadrature method is examined, for $L=4$ diffusive variables. In the interval $[\omega_{\min},\omega_{\max}]$, the error obtained with optimization is roughly 100 times smaller than with Gauss-Jacobi polynomials. Outside this interval, the optimized solution worsens logically.

In figure \ref{FigErreur}-(ii), nonlinear optimization is tested for various numbers of diffusive variables: $L=2$, 4 and 6. Improvement of the diffusive approximation as $L$ increases is observed. In counterpart, the computational cost of the numerical scheme (section \ref{SecNumSchem}) increases linearly with $L$. In practice, we will use the value $L=4$ in forthcoming experiments, which provides a relative error of model near $0.5\,\%$ in the interval of optimization. 

To conclude this section, let us mention that other choices of $\omega_{\min}$ and $\omega_{\max}$ have been tested. Logically, the accuracy of the optimization is degraded if the optimization range is increased. Nevertheless, the results remain much more accurate than those obtained with Gaussian quadrature.


\section{Exact solution of the linear fractional advection}\label{SecExact}

\subsection{Particular cases $\alpha=1/3$ and $\alpha=1/2$}\label{SecExactParti}

We consider the case of linear advection with fractional attenuation. A boundary condition is applied and the initial conditions are null. Taking $b=0$ in (\ref{ToyModel}) leads to the system
\begin{equation}
\left\{
\begin{array}{l}
\ds
\frac{\textstyle \partial u}{\partial t}+a\frac{\textstyle \partial u}{\textstyle \partial x}+\varepsilon\,D^\alpha_t u=\delta_0(x)g(t),\quad t>0,\\
[10pt]
\ds
u(x,0)=0, \quad x\in \mathbb{R}.
\end{array}
\right.
\label{ToyModelLin1}
\end{equation}
Applying a Fourier transform in space and a Laplace transform in time to (\ref{ToyModelLin1}) yields
\begin{equation}
U(k,s)=\frac{G(s)}{s+\varepsilon s^\alpha+iak},
\label{Vks}
\end{equation}
where $s$ is the Laplace variable. One defines $\lambda=(s+\varepsilon s^\alpha)/a$. Since $\real{s}>0$, then $\real{\lambda}>0$. It follows that $1/(\lambda+ik)$ is the Laplace transform of $\exp(-\lambda\,x)$ for $x>0$. Consequently, one gets
\begin{equation}
U(x,s)=\exp\left(-\frac{\varepsilon\,x}{a}\,s^\alpha\right)\,\exp\left(-\frac{x}{a}s\right)\,G(s),\hspace{0.5cm} x>0.
\label{Vxs}
\end{equation}
Setting $y=\frac{\varepsilon x}{a}$, equation (\ref{Vxs}) gives
\begin{equation}
u(x,t)=h_\alpha(y,t)\mathop{*}\limits_{t}g\left(t-\frac{x}{a}\right),\hspace{0.5cm} x>0,
\label{ExactConvol}
\end{equation}
where $h_\alpha(y,t)$ is the inverse Laplace transform of $\exp(-y\,s^\alpha)$. For $\alpha=1/3$ and $\alpha=1/2$, analytical expressions of these inverse transforms are known (see p.~120 of \cite{Mainardi10}): one has 
\begin{subnumcases}{\label{InverseLaplace}}
\ds
h_{1/3}(y,t)=\frac{y}{3^{1/3}\,t^{4/3}}\mbox{Ai}\left(\frac{y}{3^{1/3}\,t^{1/3}}\right),\label{InverseLaplace1}\\ 
[10pt]
\ds
h_{1/2}(y,t)=\frac{y}{2\,\sqrt{\pi}\,t^{3/2}}\,\exp\left(-\frac{y^2}{4\,t}\right).\label{InverseLaplace2}
\end{subnumcases}
In (\ref{InverseLaplace1}), Ai is the Airy function \cite{NRPAS}. The convolution product in (\ref{ExactConvol}) is computed numerically by the Simpson method. 


\subsection{General case}\label{SecExactGene}

\begin{figure}[htbp]
\begin{center}
\begin{tabular}{cc}
\hspace{-0.8cm}
(a) & (b) \\
\hspace{-0.8cm}
\includegraphics[scale=0.33]{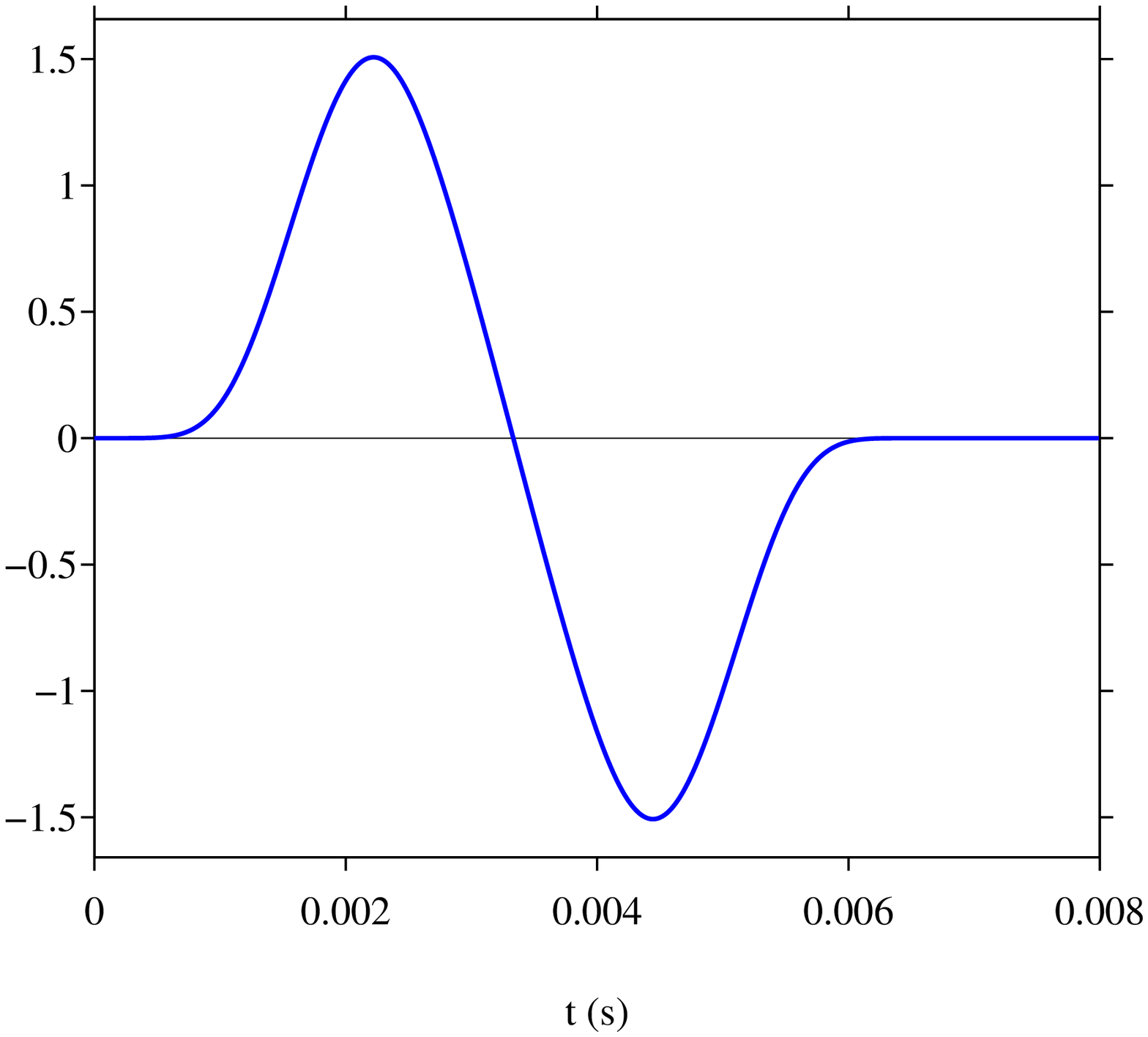}&
\hspace{-0.8cm}
\includegraphics[scale=0.33]{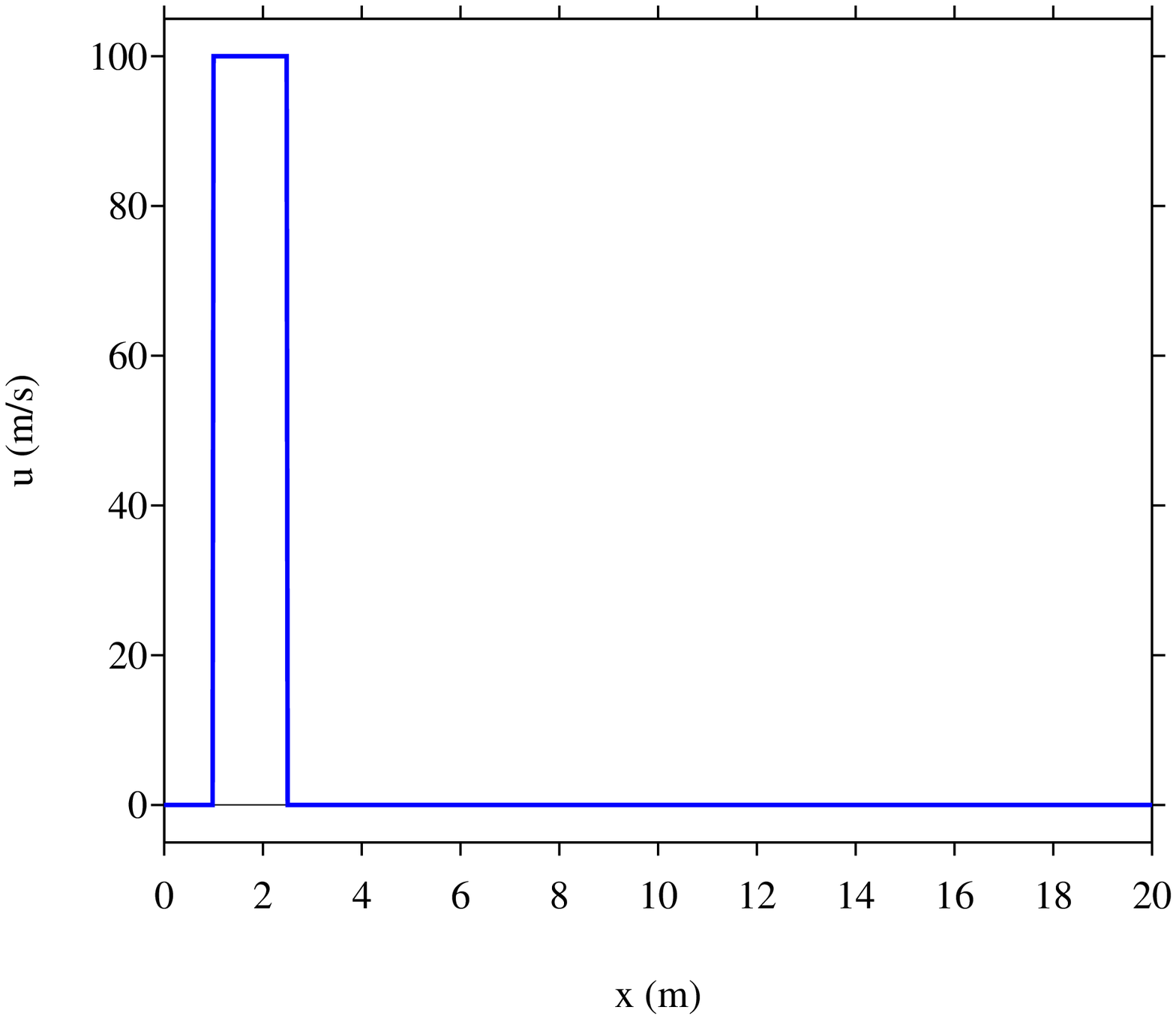} 
\end{tabular}
\vspace{-0.5cm}
\caption{\label{FigSource} Signals used in the numerical experiments. (a): time evolution of the source $g$ (\ref{JKPS_C6}). (b) spatial evolution of the rectangular pulse (\ref{Heaviside}).} 
\end{center}
\end{figure}

The exact solution detailed in section \ref{SecExactParti} is very efficient numerically. But it is restricted to particular values of the fractional order $\alpha$ and to the fractional model (\ref{ToyModel}). Here we detail an alternative approach, more tedious numerically but also more general: arbitrary values of $\alpha$ can be handled, as well as  the diffusive model (\ref{EDP}). To do so, Fourier transforms in time and space (\ref{Fourier}) are applied to (\ref{ToyModel}) or (\ref{EDP}). It gives
\begin{equation}
\hat u(k,\omega)=-i\,\frac{1}{k-k_0}\,g(\omega), \hspace{0.8cm} \mbox{ with  } \hspace{0.4cm} k_0=-\left(\frac{\omega}{a}-i\,\frac{\varepsilon}{a}\chi(\omega)\right).
\label{ExactFourierKF}
\end{equation}
$\chi$ is the symbol of the fractional PDE (\ref{ChiFD}) or the symbol of the diffusive PDE (\ref{ChiDA}). An inverse Fourier transform in space of (\ref{ExactFourierKF}) yields
\begin{equation}
u(x,\omega)=-\frac{i}{2\pi}\int_{-\infty}^{+\infty}\frac{e^{ikx}}{k-k_0}\,dk.
\label{ExactFourierXF}
\end{equation}
In the case $x>0$, the residue theorem provides
\begin{equation}
u(x,\omega)=g(\omega)\,e^{ik_0 x}.
\label{ExactFourierResidue}
\end{equation}
The inverse Fourier in time of (\ref{ExactFourierResidue}) is computed numerically by a quadrature formula on $N_f$ modes, with a frequency step $\Delta f$.


\section{Numerical results}\label{SecRes}

\subsection{Configuration}\label{SecResConf}

In all the forthcoming experiments, a domain of length 20 m is discretized on $N_x=1000$ grid nodes. Unless specified otherwise, the advection parameters are $a=300$ m/s and $b=1$. The number of memory variables is $L=4$. The CFL number is $\Upsilon=0.95$ (\ref{CFLopti}). Two times of excitation are considered: 
\begin{itemize}
\item a source term. The time evolution is a truncated combination of sinusoids with $C^6$ smoothness:
\begin{equation}
g(t) = 
\left\{
\begin{array}{l}
\ds
V\,\sum_{m=1}^4 a_m \sin\,(b_m\,\omega_c\,t)  \mbox{  if }\;0\leq t\leq \frac{1}{f_c},\\
[8pt]
\ds
0 \qquad \mbox{otherwise},
\end{array}
\right. 
\label{JKPS_C6}
\end{equation}
with parameters $b_m=2^{m-1}$, $a_1=1$, $a_2=-21/32$, $a_3=63/768$ and $a_4=-1/512$. The central frequency is $f_c=150$ Hz;
\item an initial condition. The space evolution is a rectangular force pulse:
\begin{equation}
u_0(x)=V\left(H\left(x-x_0\right)-H\left(x-x_0-\lambda\right)\right),
\label{Heaviside}
\end{equation}
where $H$ is the Heaviside function, $x_0=1$ m, and $\lambda=1.5$ m.
\end{itemize} 
These excitations are displayed in figure \ref{FigSource}.

\begin{figure}[htbp]
\begin{center}
\begin{tabular}{cc}
\hspace{-0.8cm}
(a) & (b) \\
\hspace{-0.8cm}
\includegraphics[scale=0.33]{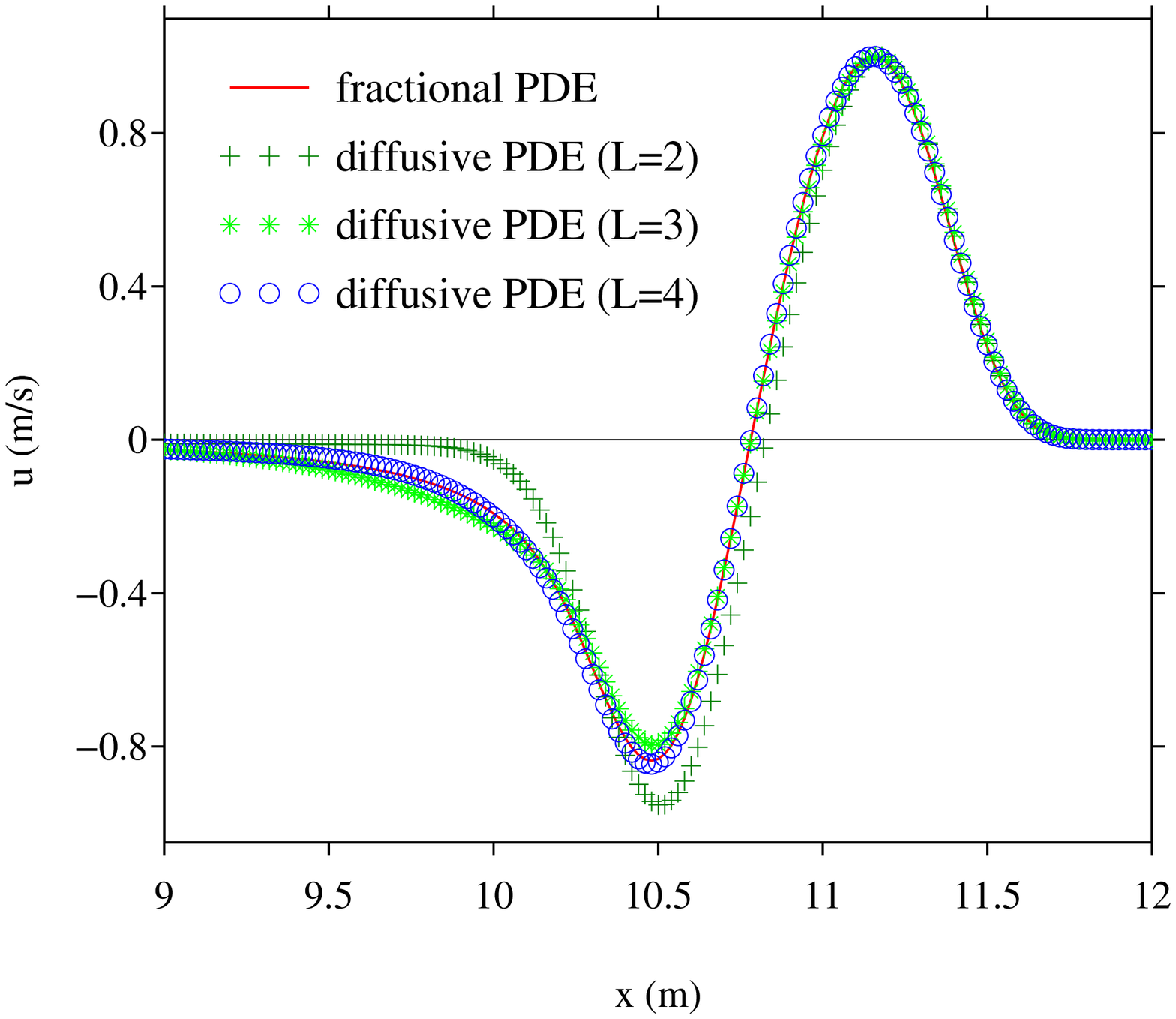}&
\hspace{-0.8cm}
\includegraphics[scale=0.33]{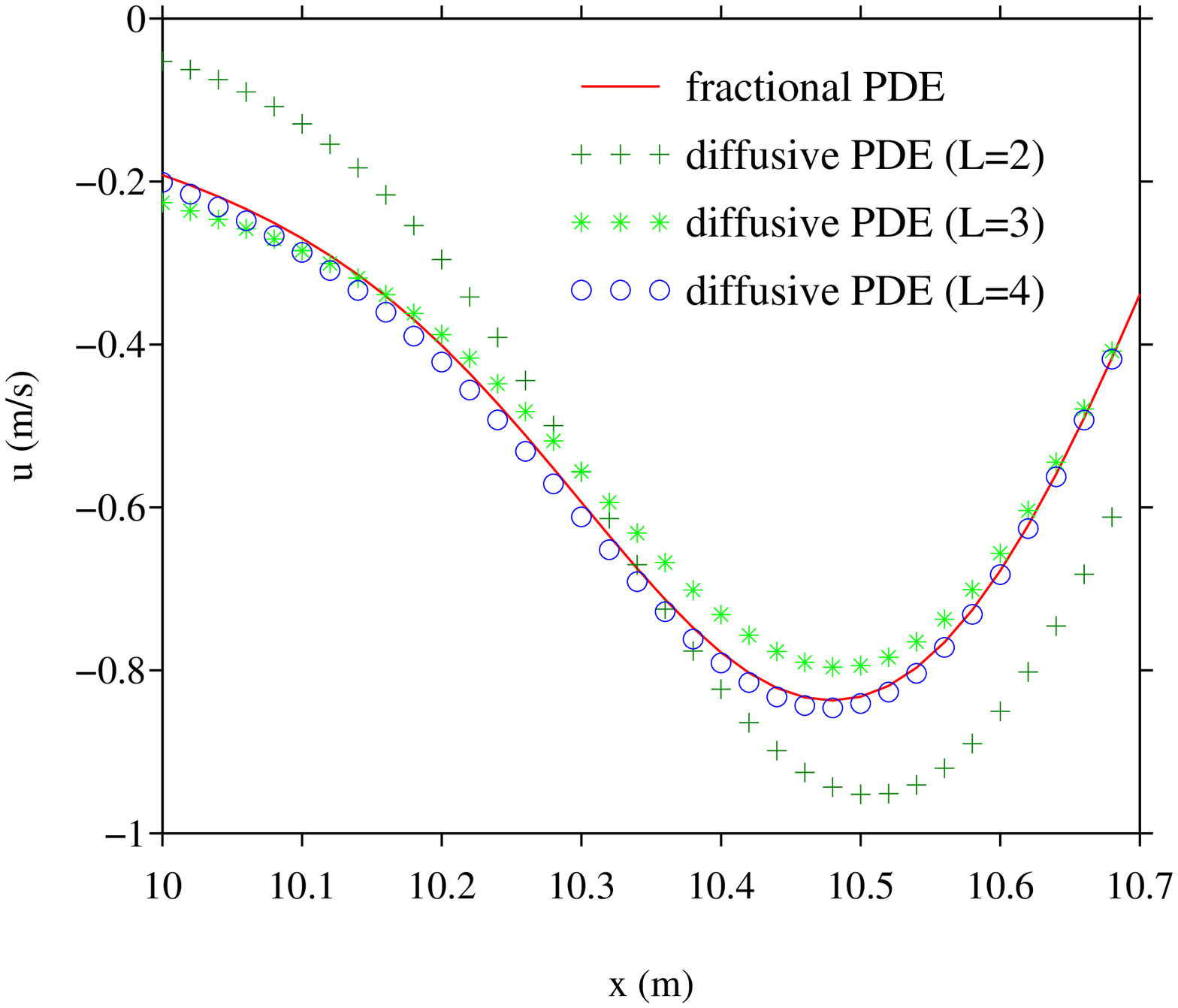} 
\end{tabular}
\vspace{-0.5cm}
\caption{\label{FigTest1-DA} Test 1: exact solutions of the fractional PDE (\ref{ToyModel}) and of the diffusive PDE (\ref{EDP}) with $\alpha=1/3$ at $t=0.04$ s (a) Zoom around one extremum of the wave (b).} 
\end{center}
\end{figure}


\subsection{Linear fractional advection}\label{SecResLinear}

\begin{figure}[htbp]
\begin{center}
\begin{tabular}{cc}
\hspace{-0.8cm}
$\alpha=1/3$ & $\alpha=1/3$ \\
\hspace{-0.8cm}
\includegraphics[scale=0.33]{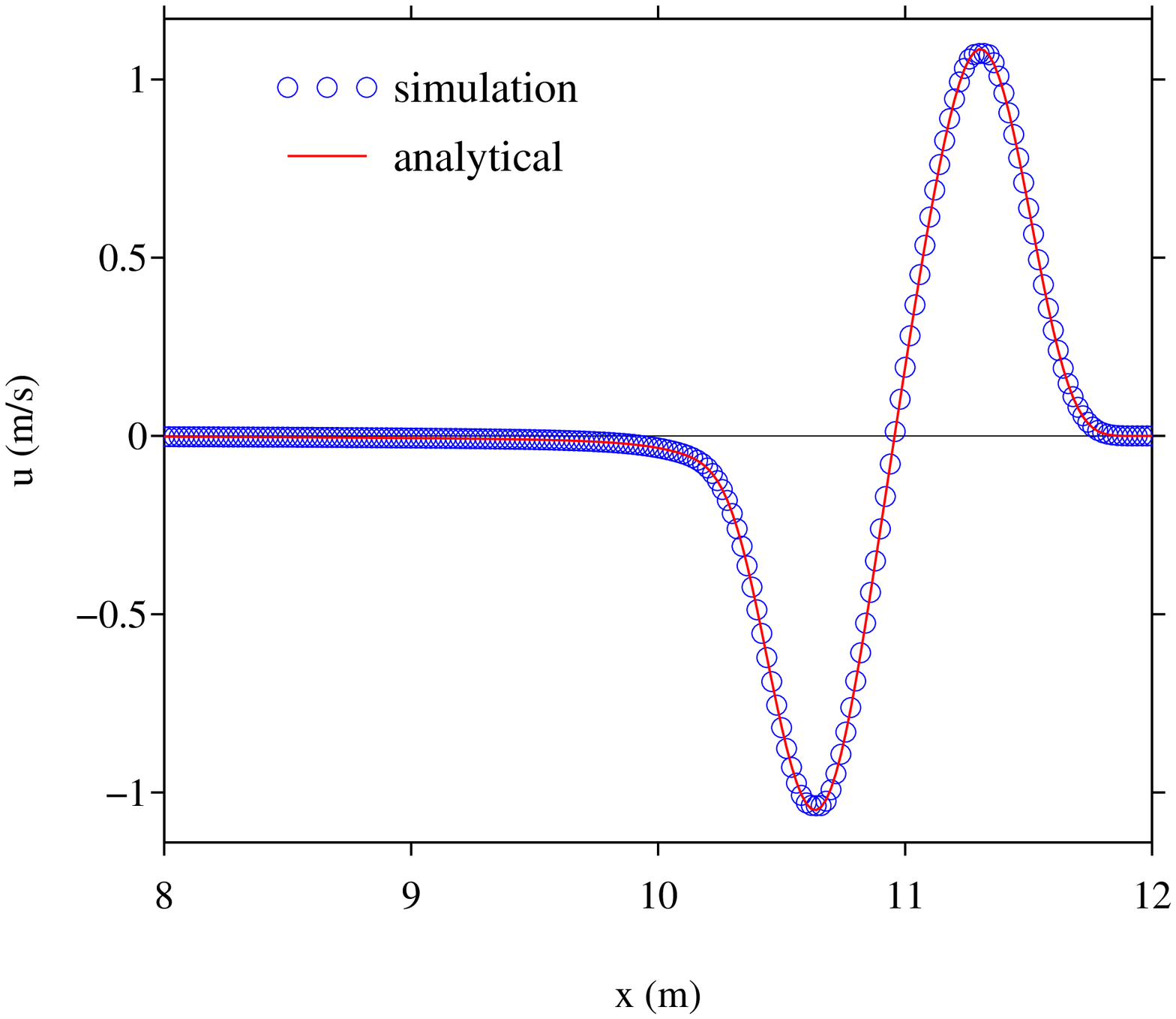}&
\hspace{-0.8cm}
\includegraphics[scale=0.33]{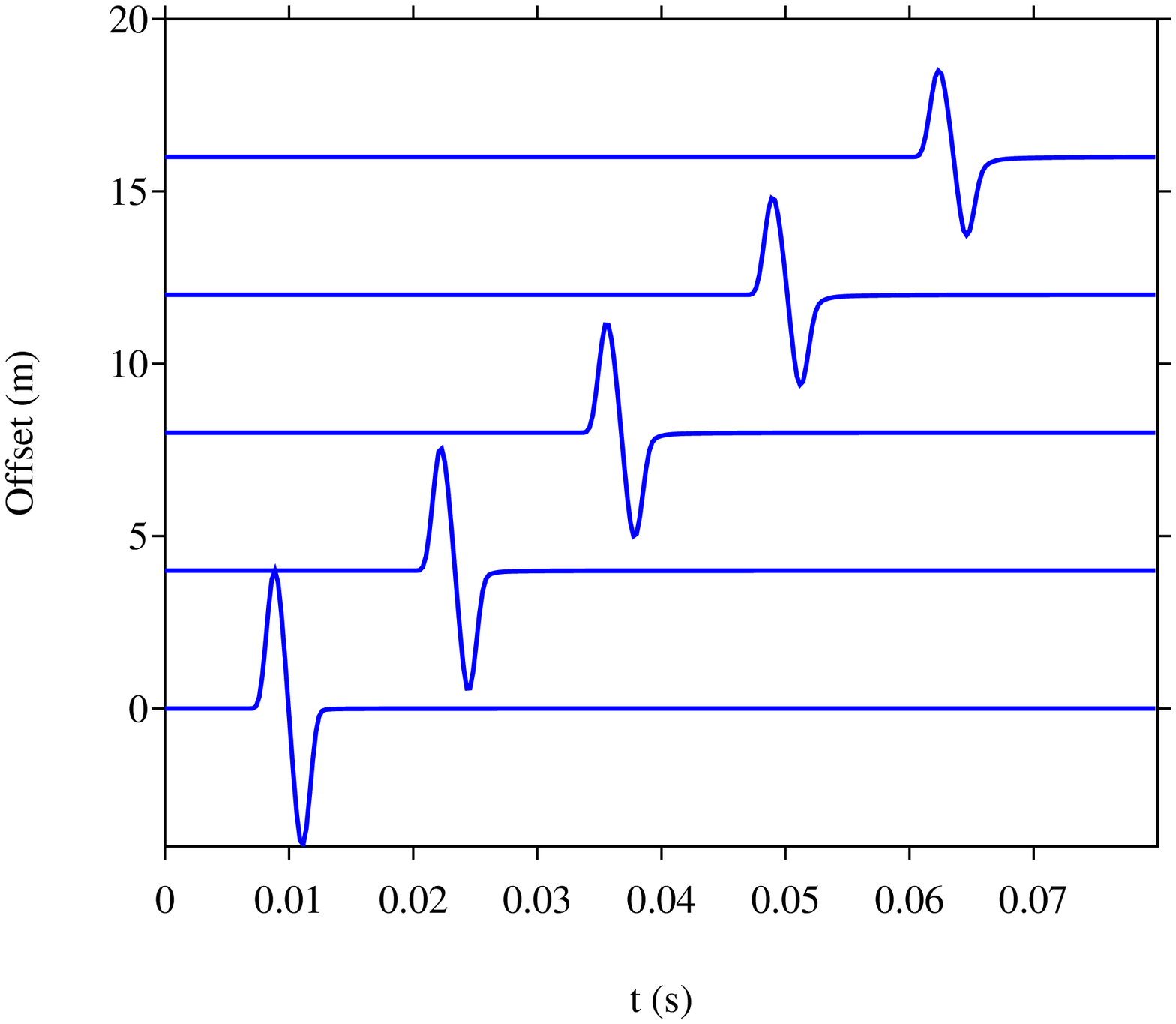} \\
$\alpha=1/2$ & $\alpha=1/2$ \\
\hspace{-0.8cm}
\includegraphics[scale=0.33]{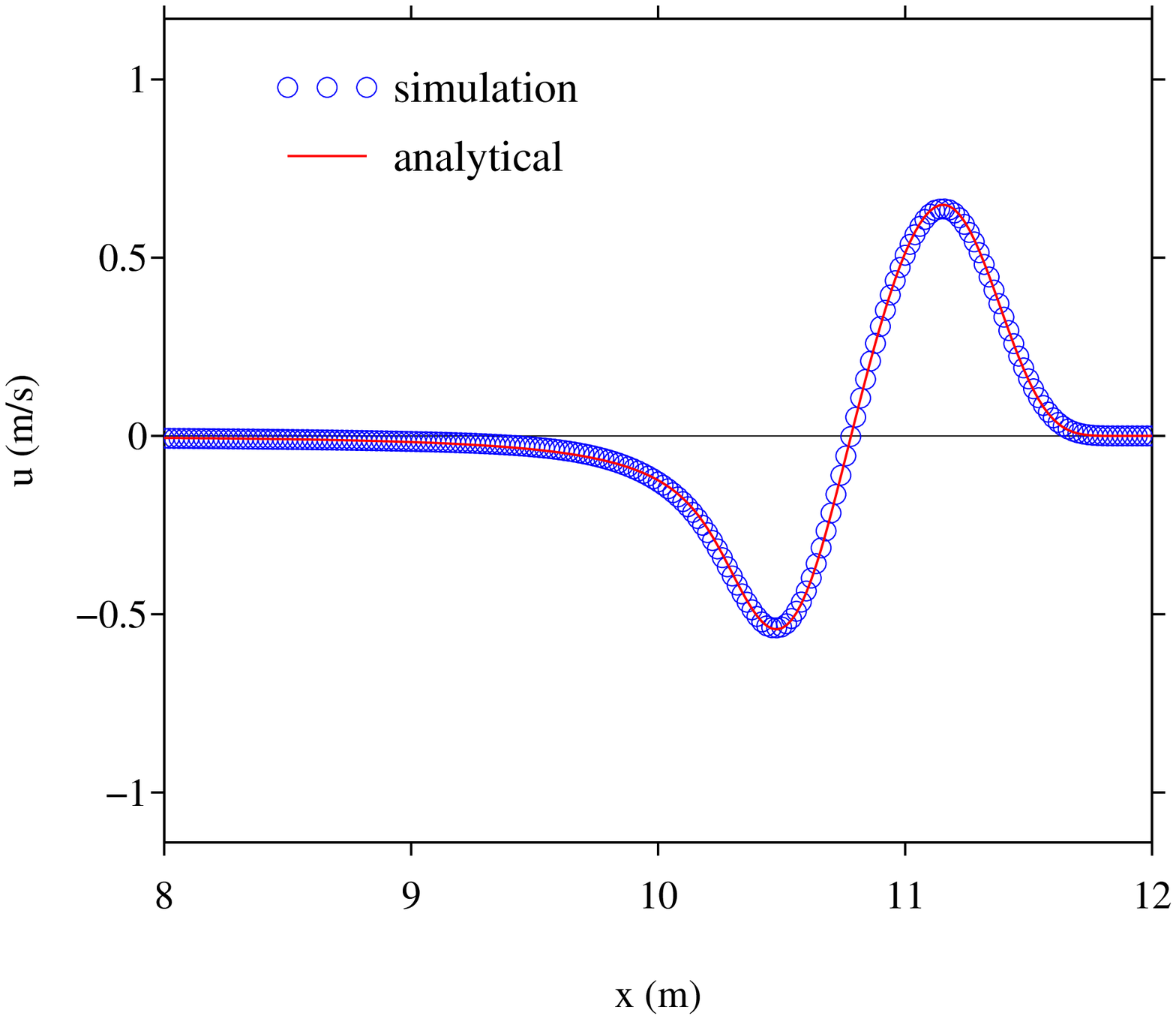}&
\hspace{-0.8cm}
\includegraphics[scale=0.33]{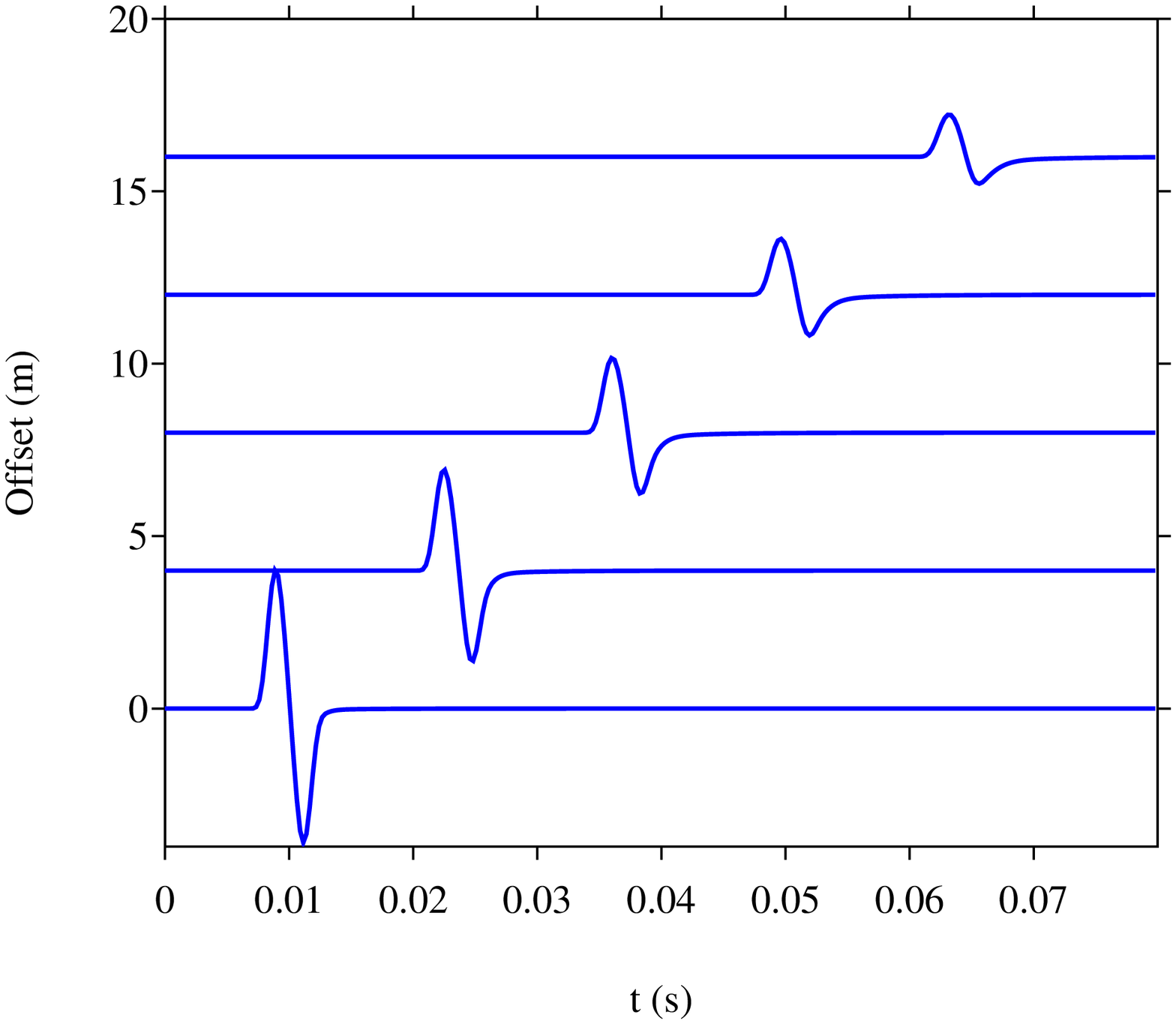} \\
$\alpha=0.7$ & $\alpha=0.7$ \\
\hspace{-0.8cm}
\includegraphics[scale=0.33]{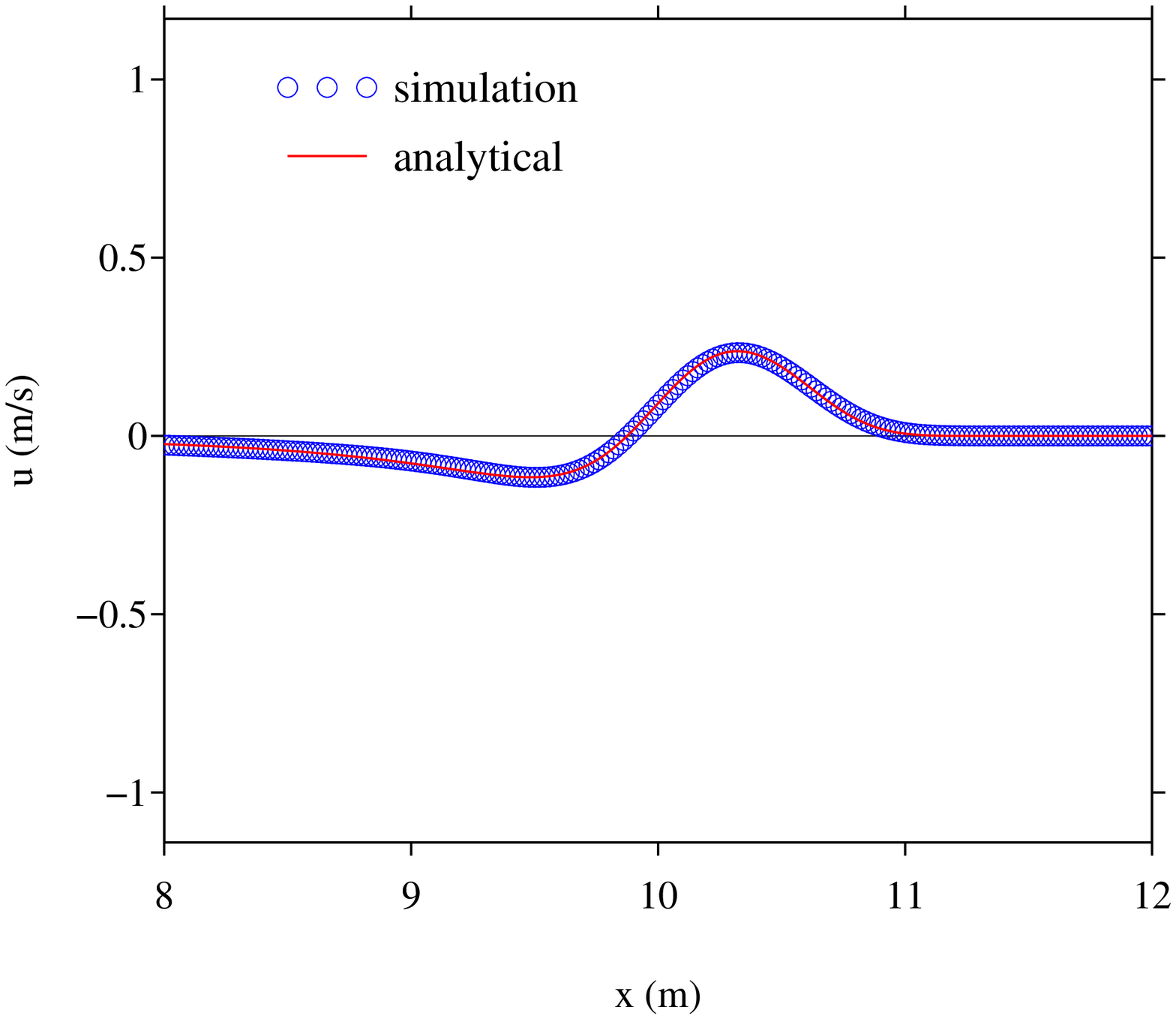}&
\hspace{-0.8cm}
\includegraphics[scale=0.33]{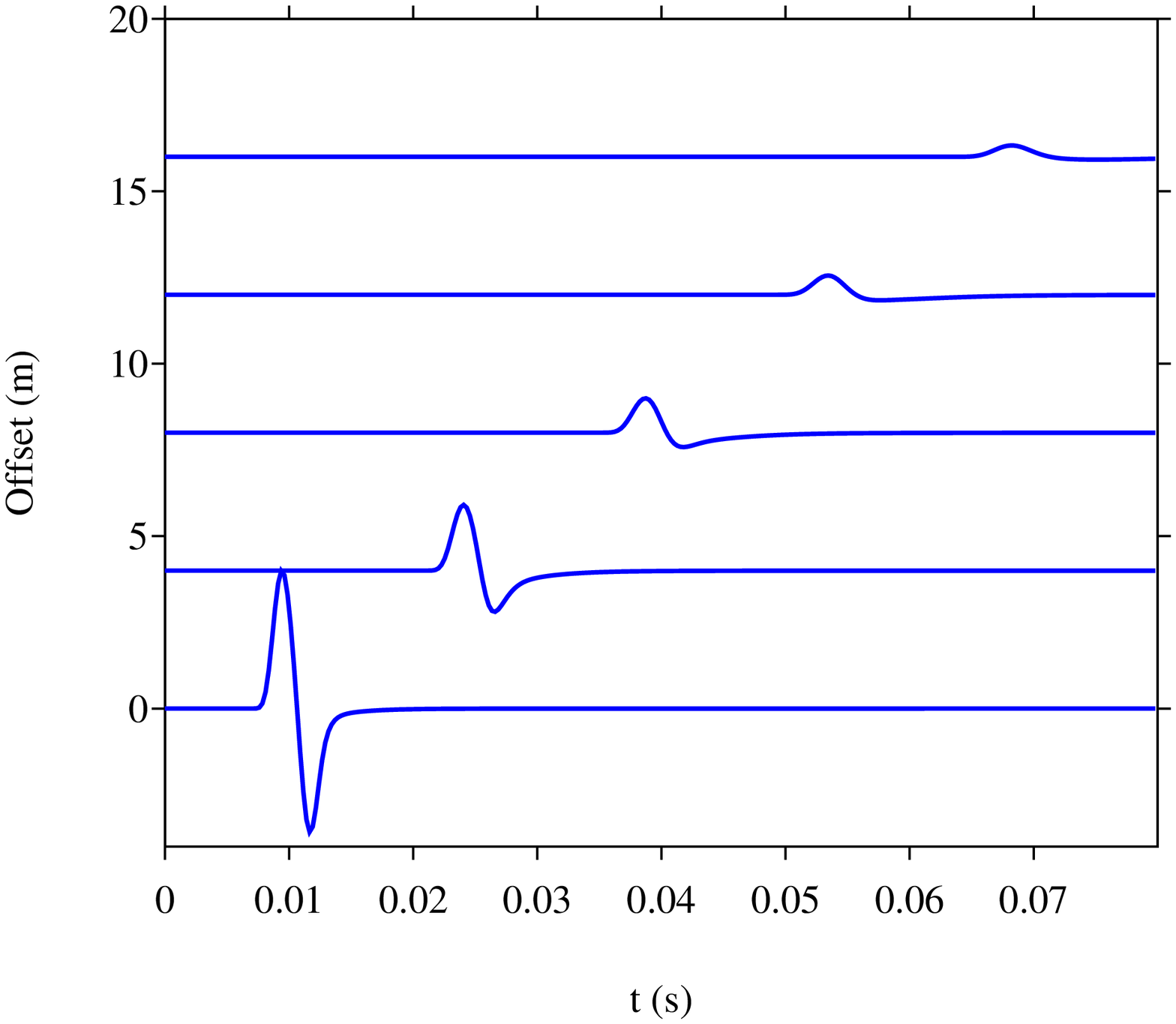} \\
\end{tabular}
\vspace{-0.5cm}
\caption{\label{FigTest1-Sismo} Test 1: time-domain simulations of linear fractional advection, for various orders of the fractional order $\alpha$. Left row: snapshots of the numerical and exact solutions. Right row: simulated seismograms.} 
\end{center}
\end{figure}

The objective of the first test is to validate both the diffusive approximation of the fractional derivative and the numerical strategy. For this purpose, the nonlinearity is neglected, and one tackles pure advection with fractional dissipation: $a=300$ m/s, $b=0$ and $\varepsilon=1$ s$^{\alpha-1}$. The initial conditions are null. The source (\ref{JKPS_C6}) is put at the left boundary of the domain, with an amplitude $V=1$ m/s. 

The figure \ref{FigTest1-DA} illustrates the effect of the number of memory variables on the approximation of the fractional PDE, for $\alpha=1/3$. For this purpose, one displays the exact solution of the fractional PDE (\ref{ToyModel}), and the exact solutions of the diffusive PDE (\ref{EDP}) for various values of $L$. See section \ref{SecExact} for details; the number of Fourier modes is $N_f=2048$, with a frequency step $\Delta f=0.75$ Hz. When $L=2$, a large error of modeling is introduced. On the contrary, $L=4$ ensures a very accurate approximation of the fractional model. This value will be used from now.

To see the effect of the fractional order on wave propagation, one considers then three values of $\alpha$: 1/3, 1/2 and 0.7. The cases $\alpha=1/3$ and $\alpha=1/2$ yield closed-form exact solutions (section \ref{SecExactParti}), whereas $\alpha=0.7$ yields a semi-analytical solution based on Fourier analysis (section \ref{SecExactGene}). The left row of figure \ref{FigTest1-Sismo} compares the numerical and analytical solutions of the diffusive PDE (\ref{EDP}) at $t=0.04$ s, which amounts to 632 time steps. The same vertical scale is used in the three cases to examplify the effect of the fractional order: as predicted by the dispersion analysis (section \ref{SecPbDisp}), the attenuation increases with $\alpha$. On the contrary, the wavefront propagates faster for small values of $\alpha$, because the phase velocity decreases with $\alpha$. In all cases, agreement is obtained between numerical and exact values. 

The right row of figure \ref{FigTest1-Sismo} displays the time evolution of $u$ at the receivers located at $x_r=2+4\,(j-1)$, with $j=1,\,\cdots, 5$, up to $t=0.08$ s (1264 time steps). The vertical scale is chosen so that the maximal value of $u$ at the first receiver (offset 0 m) reaches the null value of $u$ at the second receiver (offset 4 m). In the case $\alpha=1/3$, a small decrease of amplitude is observed during the propagation. In the case $\alpha=1/2$, the attenuation increases, and the wave is more dispersed. The arrival time at the last receiver (offset 16 m) is greater than in the case $\alpha=1/3$, which means that the velocity is smaller when $\alpha$ is greater, as predicted theoretically. These observations are even more pronounced in the case $\alpha=0.7$, where the signal at offset 16 m has almost disappeared at the scale of the figure. The behaviors illustrated on these seismograms are similar to those observed in viscoelasticity, where typical models of attenuation involve quality factors $Q(\omega)\sim Q_0\omega^{-\alpha}$ \cite{Blanc16}.

 
\subsection{Nonlinear advection}\label{SecResNonlinear} 
 
\begin{figure}[htbp]
\begin{center}
\begin{tabular}{cc}
\hspace{-0.8cm}
$\varepsilon=0$, $t_1$ (a) & $\varepsilon=0$, $t_2>t_1$ (b) \\
\hspace{-0.8cm}
\includegraphics[scale=0.33]{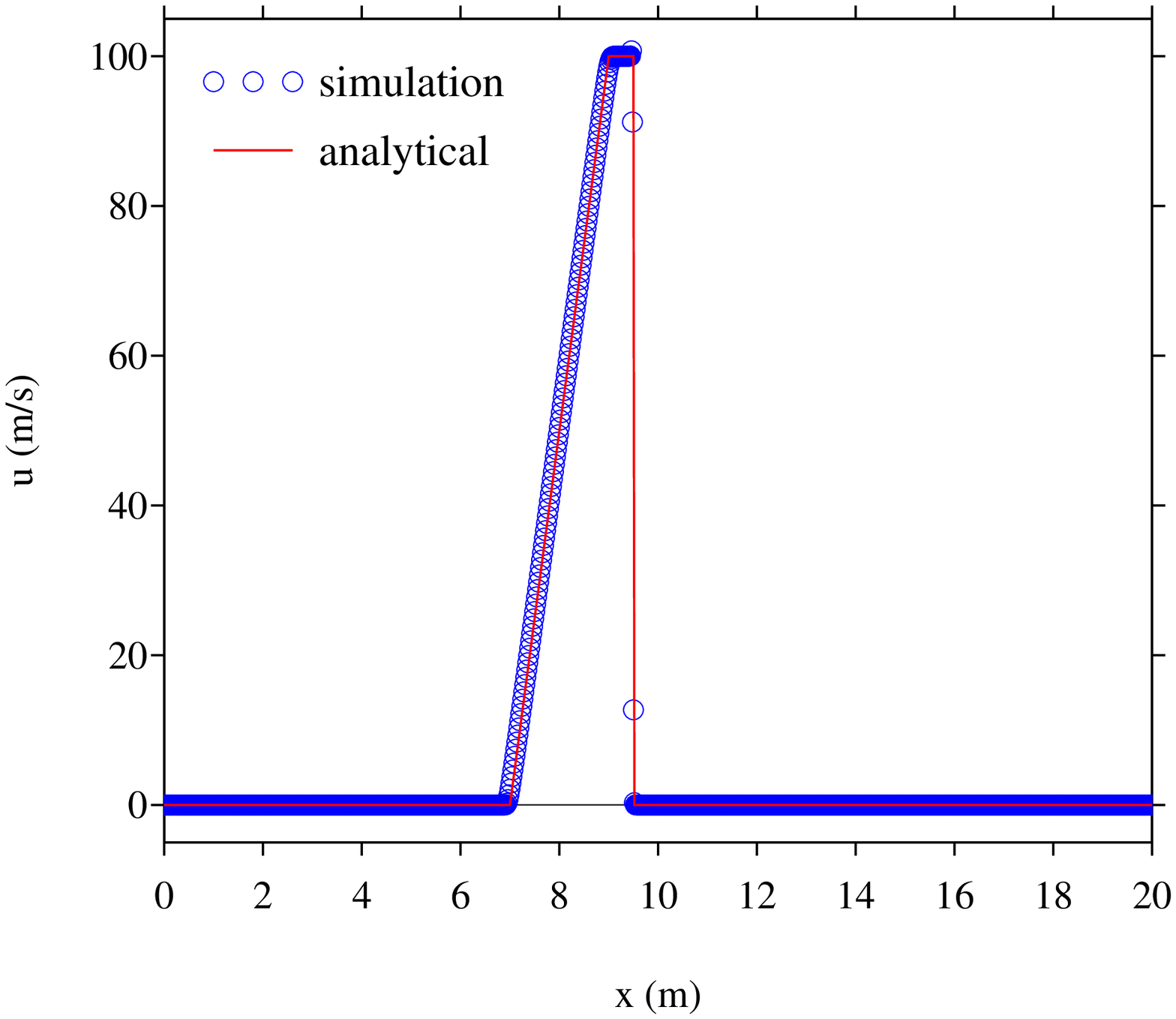}&
\hspace{-0.8cm}
\includegraphics[scale=0.33]{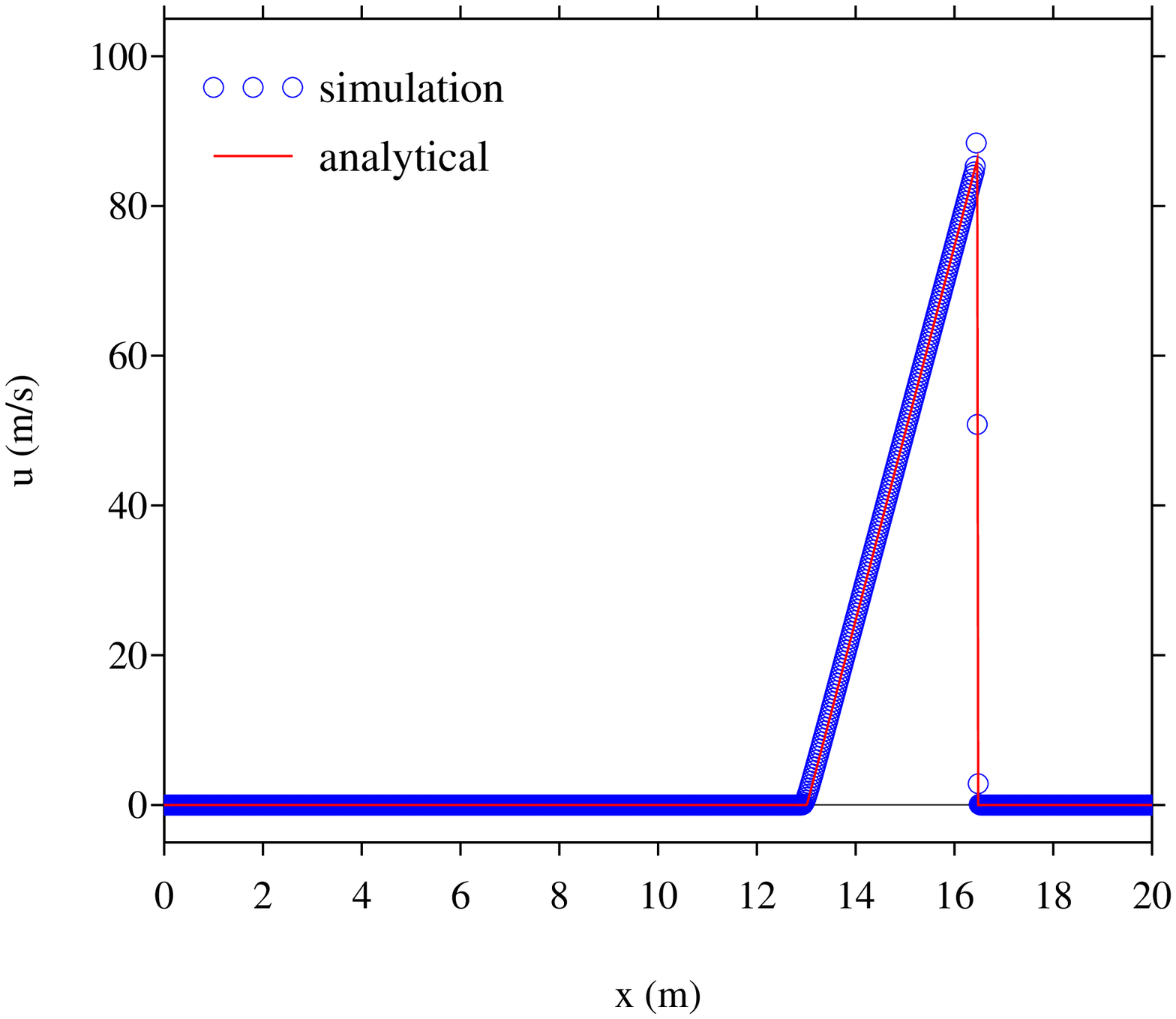} \\
\hspace{-0.8cm}
$\alpha=1/2$, $t_1$ (c) & $\alpha=1/2$, $t_2>t_1$ (d) \\
\hspace{-0.8cm}
\includegraphics[scale=0.33]{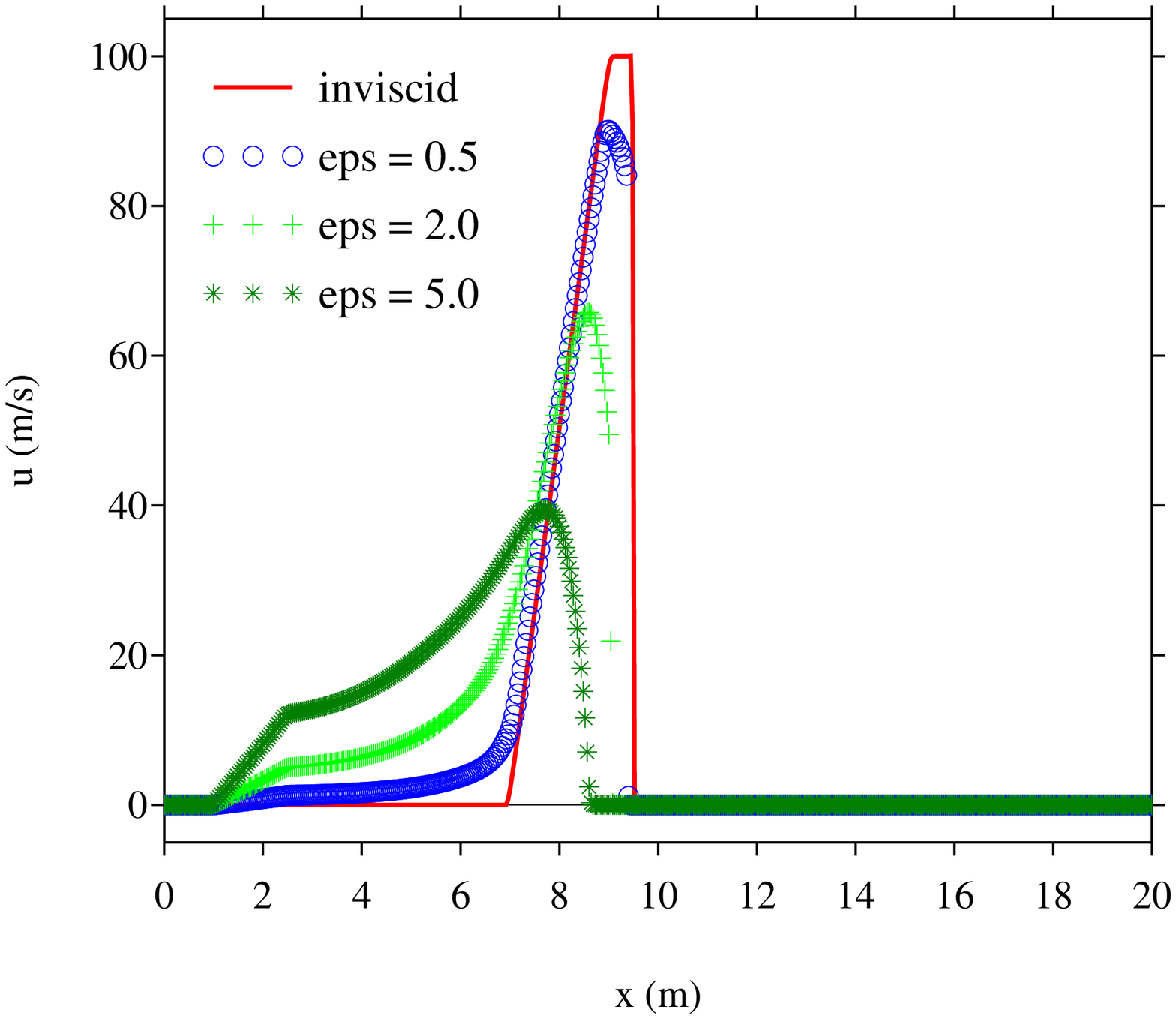}&
\hspace{-0.8cm}
\includegraphics[scale=0.33]{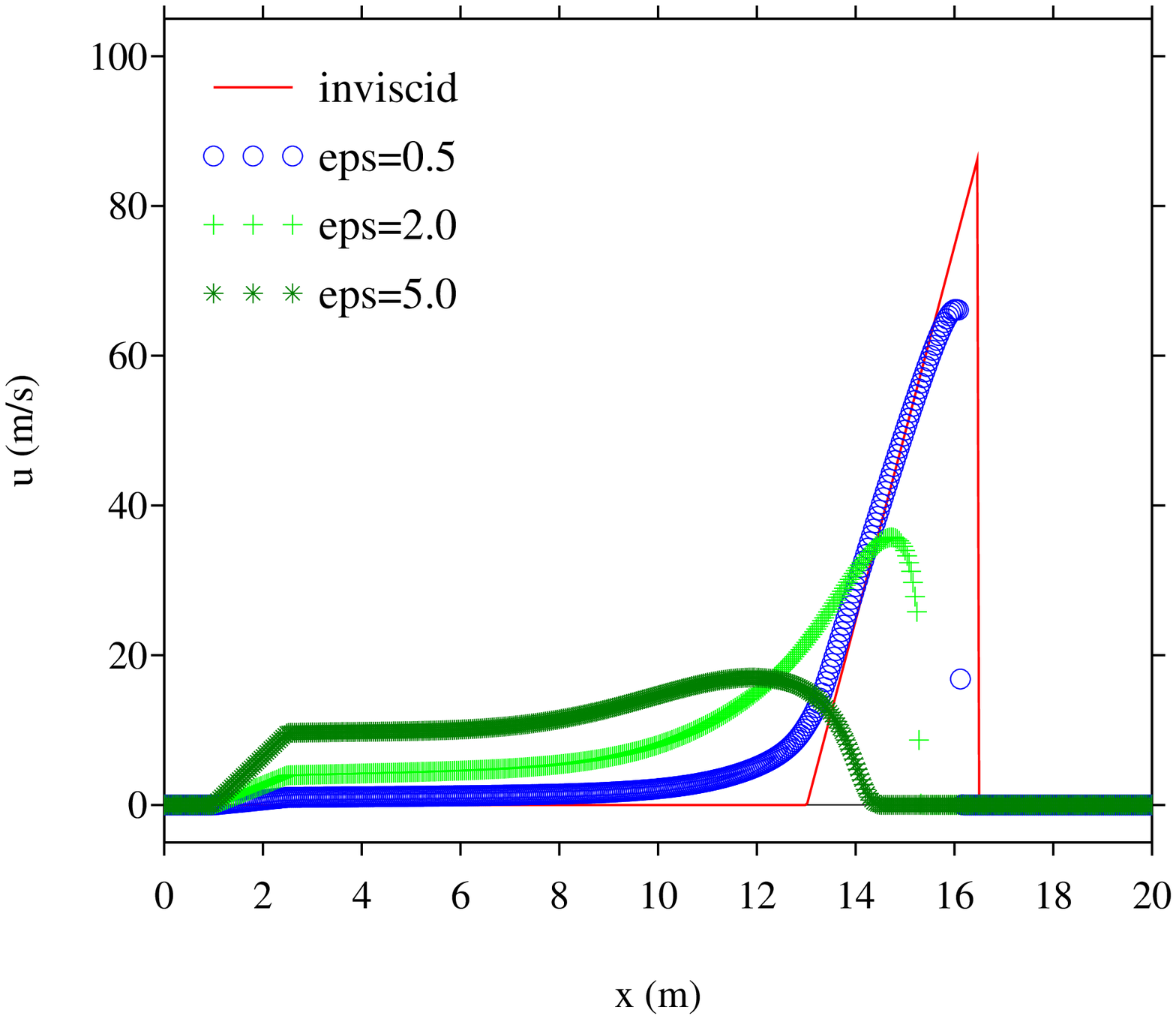} \\
\hspace{-0.8cm}
$\varepsilon=2$, $t_1$ (e) & $\varepsilon=2$, $t_2>t_1$ (f) \\
\hspace{-0.8cm}
\includegraphics[scale=0.33]{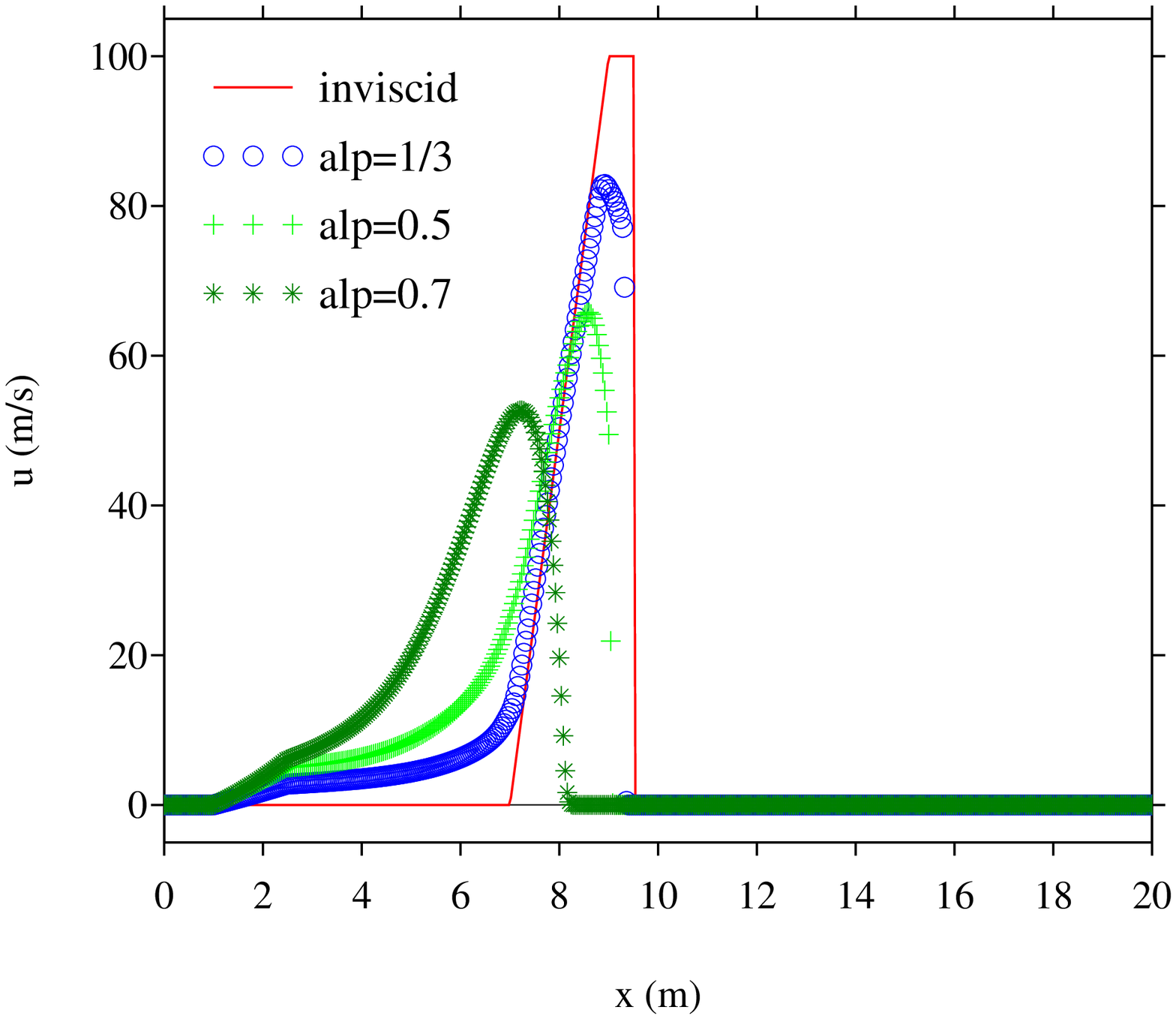}&
\hspace{-0.8cm}
\includegraphics[scale=0.33]{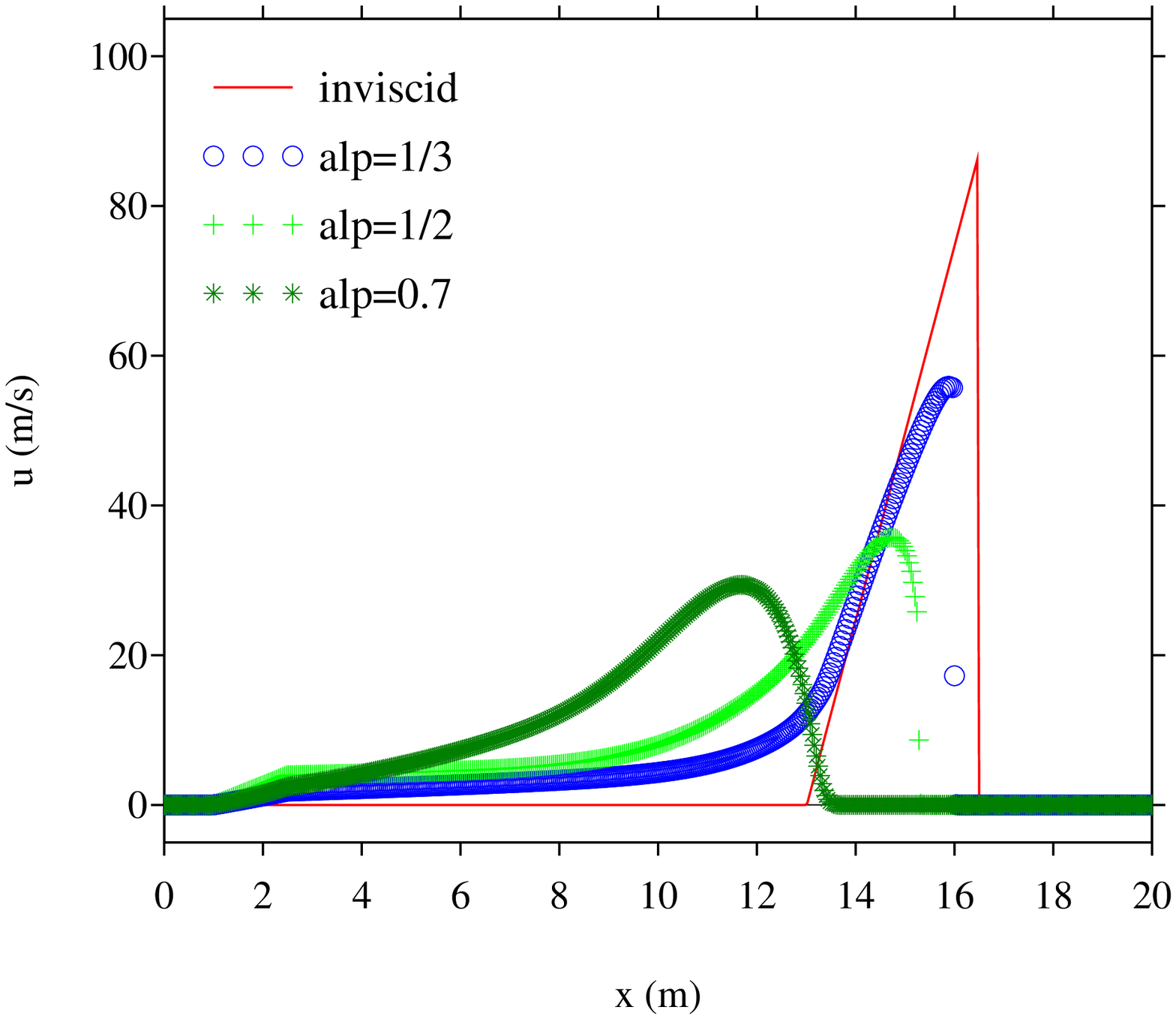} \\
\end{tabular}
\vspace{-0.5cm}
\caption{\label{FigTest2} Test 2: snapshots of the solution at $t_1$ (left row) and $t_2>t_1$ (right row). Nonlinear advection and fractional attenuation are considered. (a-b): numerical and exact solutions without attenuation. (c-d): numerical solutions for $\alpha=1/2$ and various values of the fractional amplitude $\varepsilon$. (e-f): numerical solutions for $\varepsilon=2$ and various values of the fractional order $\alpha$. In (c-f), the inviscid case is computed analytically.} 
\end{center}
\end{figure}

The aim of the second test is to see the effect of fractional attenuation on existing discontinuities. The coefficients of the nonlinear wave propagation are $a=300$ m/s and $b=1$. Various values of the fractional parameters are considered: $\varepsilon=0$ (no attenuation), 0.5, 2 and 5; $\alpha=1/3$, 1/2 and 0.7. Contrary to test 1, no forcing term is considered: $g(t)=0$. The computations are initialized by the rectangular force pulse (\ref{Heaviside}), with $V=100$ m/s. Figure \ref{FigTest2} displays the numerical solution at $t_1=0.02$ s (left row) and at $t_2=0.04$ s (right row). 

In the absence of attenuation ($\varepsilon=0$), the solution is known analytically. Left and right parts of the initial pulse yield a rarefaction wave and a shock wave, respectively (a). At $t^*=0.03$ s, the rarefaction reaches the shock; then, the shock velocity and the amplitude decrease (b). The comparisons between numerical and analytical solutions confirm that the nonlinear wave propagation is correctly simulated.

When $\varepsilon \neq 0$, no analytical solution is known. Figure \ref{FigTest2}-(c),(d) display the numerical solutions when $\alpha=1/2$, for various amplitudes of $\varepsilon$. As predicted by the dispersion analysis (section \ref{SecPbDisp}), the phase velocity and the amplitude of the signal decrease when $\varepsilon$ increases. For small $\varepsilon$, the shock seems to be maintained. For greater values ($\varepsilon=2$ and 5), the shock disappears and is smeared. Similar conclusions are obtained at a given $\varepsilon$ and for increasing values of $\alpha$, as displayed on figure \ref{FigTest2}-(e),(f). 


\subsection{Occurence of shocks}\label{SecResShock}

As a last experiment, we examine the emergence of shocks in the fractional Burger's equation, when a smooth source is injected. The initial conditions are null. A source is excited at the left boundary, with the time evolution (\ref{JKPS_C6}) and the amplitude $V=20$ m/s. 

Figure \ref{FigTest3} displays the snapshots of the numerical solutions at $t=0.06$ s. As in test 2, one observes the effect of increasing values of $\varepsilon$ and $\alpha$: increase of attenuation and decrease of velocity, as predicted by the dispersion analysis. Without attenuation ($\varepsilon=0$), shocks have emerged, leading to classical sawtooth waveforms. For small values of $\varepsilon$ and $\alpha$, the sharp fronts seem to be maintained: only one grid node lies in the sharp profile, probably due to the numerical attenuation. But for higher values of the fractional parameters, the profiles are smeared and the sharp fronts disappear. Contrary to the inviscid case, these simulations indicate that the emergence of discontinuities in the fractional Burger's equation is conditional.

It is emphasized that these numerical experiments are only indications. Indeed, the integration of the hyperbolic step introduces numerical smearing, and one must be cautious when interpreting a waveform as a shock or not. Our goal here is only to motivate further mathematical analysis.

\begin{figure}[htbp]
\begin{center}
\begin{tabular}{cc}
\hspace{-0.8cm}
$\alpha=1/2$ (a) & $\varepsilon=1$ (b) \\
\hspace{-0.8cm}
\includegraphics[scale=0.33]{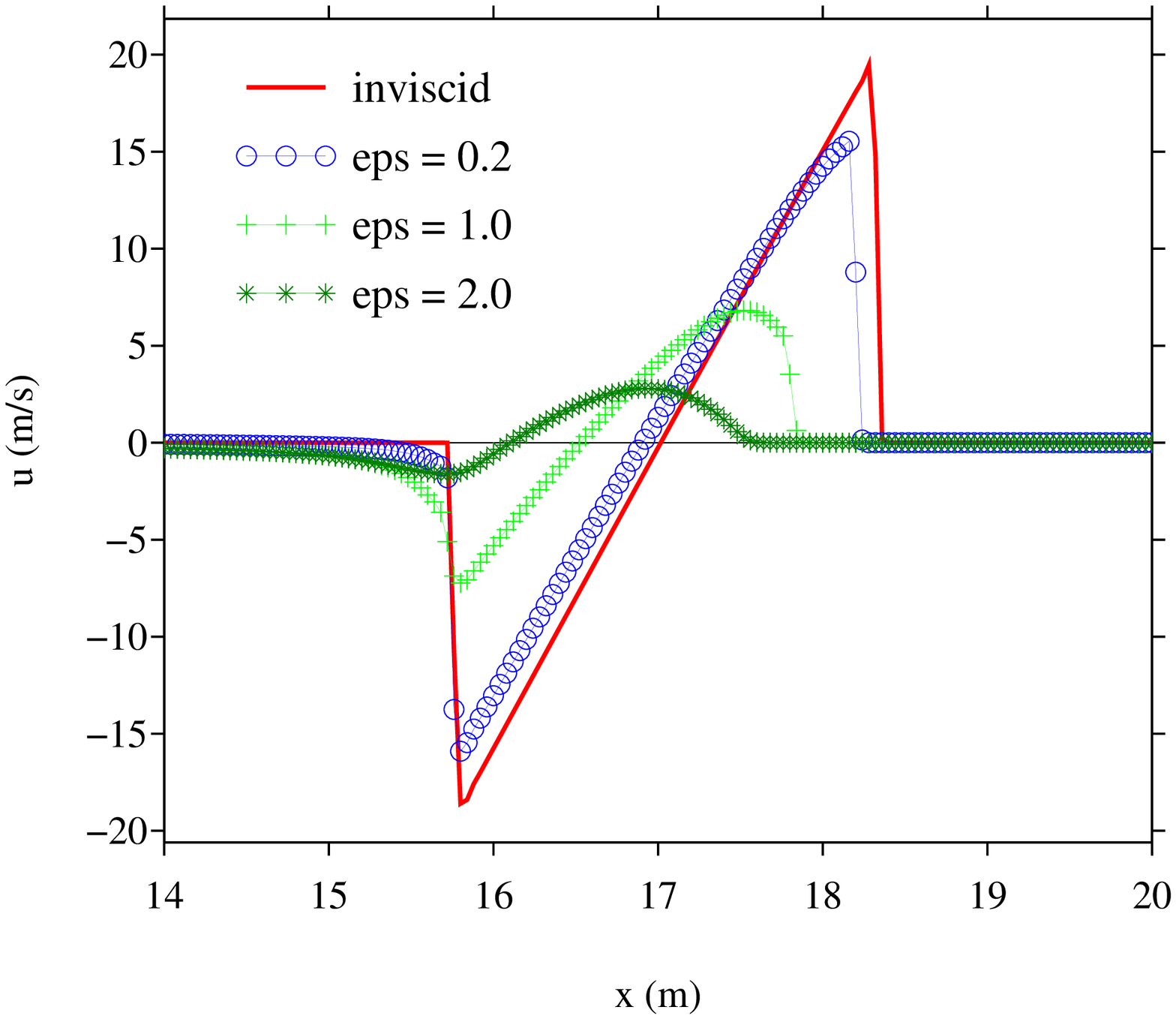}&
\hspace{-0.8cm}
\includegraphics[scale=0.33]{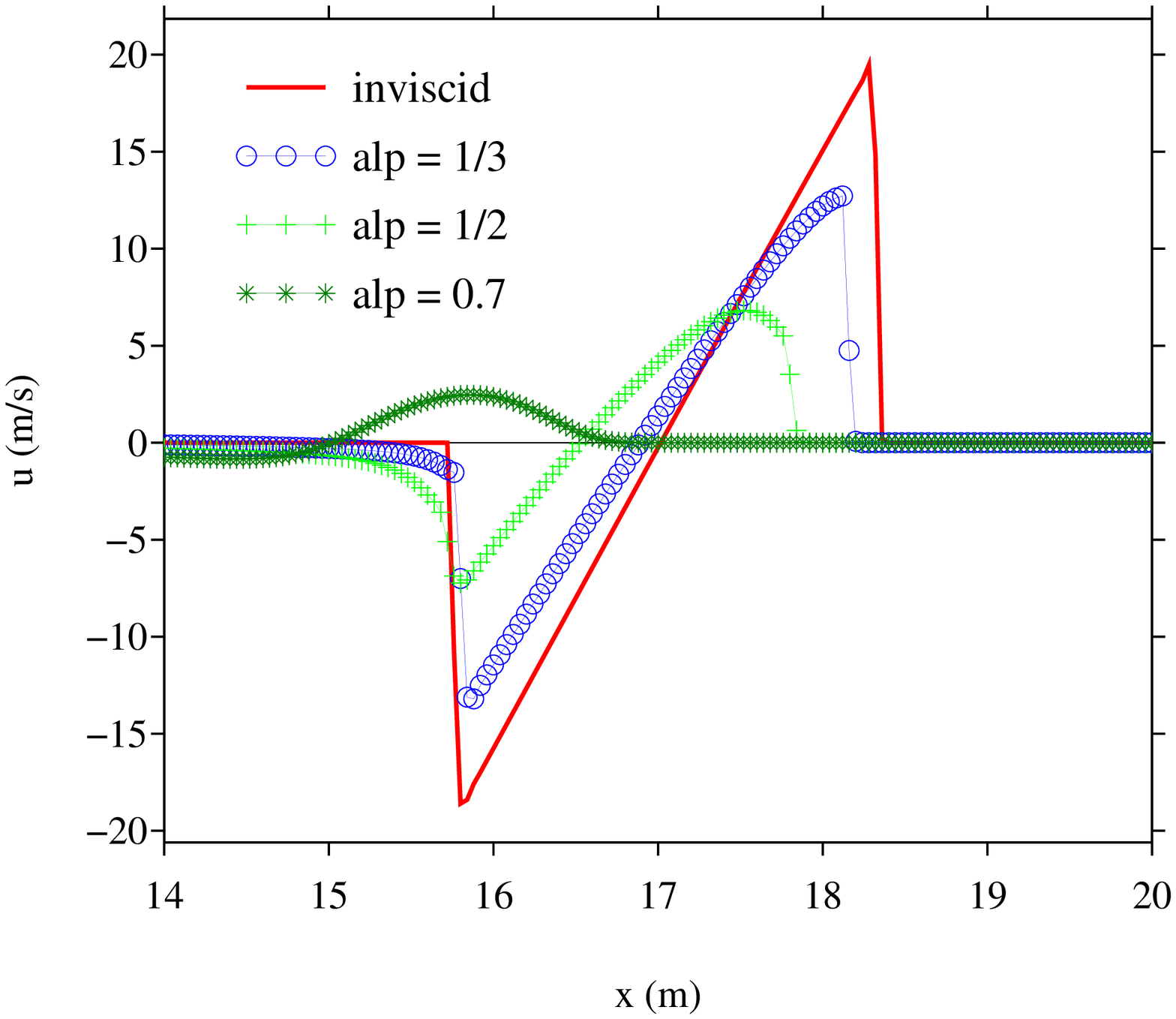} 
\end{tabular}
\vspace{-0.5cm}
\caption{\label{FigTest3} Test 3: snapshots of the numerical solutions at $t=0.06$ s. Nonlinear advection and fractional attenuation are considered. (a) $\alpha=1/2$, and various values of $\varepsilon$. (b) $\varepsilon=1$, and various values of $\alpha$.} 
\end{center}
\end{figure}


\section{Conclusion}\label{SecConclu}

We have proposed a numerical strategy to solve a nonlocal nonlinear hyperbolic equation with fractional attenuation (\ref{TM}). This approach requires to introduce some memory variables to keep track of the past of the solution. In counterpart, the model obtained is well-suited to numerical discretization. The condition of stability is not modified compared with the inviscid hyperbolic equation, and the discrete energy decreases. Moreover, an optimized characterization of the memory variables greatly reduces the number of arrays.

Another strategy is commonly used in nonlinear acoustics, based on a mixed resolution: the propagative part is solved in the space-time domain, whereas the fractional derivative is solved in the space-frequency domain \cite{Dagrau11,Giammarinaro16}. The use of diffusive approximation in the community of nonlinear acoustics thus requires a detailed comparison of efficiency between these two approaches. Concretely, the CPU-time of direct and inverse FFT should be compared to the one induced by the matrix-vector products (\ref{SplitDiffuExp}), at each time and step increment. This analysis is left for future studies.

This article is an attempt for better understanding the competition between nonlinear effects and nonlocal relaxation. Many theoretical questions remain to be addressed. In particular, the numerical experiments have raised the question of regularity of the solutions. Unlike the inviscid Burger's equations, it seems that the emergence of shocks is conditional, as in Burger's equation with linear source term
\begin{equation}
\frac{\textstyle \partial u}{\partial t}+\frac{\textstyle \partial}{\textstyle \partial x}\left(a\,u+b\,\frac{\textstyle u^2}{\textstyle 2}\right)u=-\varepsilon u,\quad \varepsilon>0.
\end{equation}
This question requires a deeper analysis to confirm / infirm the numerical observations. The exact solution of the Riemann problem needs also to be computed.

On the other hand, a similar approach could be adapted to a larger class of hyperbolic equations of the form
\begin{equation}
\frac{\textstyle \partial u}{\partial t}+\frac{\textstyle \partial}{\textstyle \partial x}f(u)+\varepsilon\,D^\alpha_t u=0.
\label{HypFrac}
\end{equation}
The scalar case of a cubic flux function $f(u)=u^3$ is of particular interest \cite{Hayes97}. It models focused acoustic beams of shear waves in soft solids, such as the grey matter in the brain \cite{Giammarinaro16}. 

More generally, the diffusive approach can be applied to a wide range of pseudo-differential time operators with a hereditary behavior. Examples may be found in mechanics for the modeling of viscoelasticity \cite{Deu10} and poroelasticity \cite{TheseBlanc,Blanc13a,Blanc16}. Other models can be investigated in electromagnetism, to describe dispersive media; see \cite{Petropoulos05} and references therein for a review. Lastly, non-hyperbolic equations with time fractional derivatives could also be investigated by applying the diffusive representation: one can think for instance to the nonlinear Erd\'ely-Kober equations describing abnormal diffusion in porous media \cite{Pagnini12,Plociniczak14}.


\appendix

\section{Chester's equation}\label{AppChester}

Equation (\ref{ToyModel}) models various nonlinear and thermoviscous wave phenomena \cite{Makarov97}. It is also related to a well-known model of finite-amplitude sound waves in a tube, as shown here. Let us assume weak nonlinearity, which means that the nonlinear term can be neglected in the expansion of $\frac{ \partial u}{\partial t}$ \cite{Hamilton98}. Based on (\ref{ToyModel}) and (\ref{Caputo}), one has: 
\begin{equation}
\begin{array}{lll}
\displaystyle
\frac{\textstyle \partial u}{\partial t}&=& \displaystyle
-\frac{\textstyle \partial}{\textstyle \partial x}\left(a\,u+b\,\frac{\textstyle u^2}{\textstyle 2}\right)-\varepsilon\,I_t^{1-\alpha}\left(\frac{\textstyle \partial u}{\partial t}\right),\\
[8pt]
&\approx& \displaystyle
-\frac{\textstyle \partial}{\textstyle \partial x}\left(a\,u+b\,\frac{\textstyle u^2}{\textstyle 2}\right)-\varepsilon\,I_t^{1-\alpha}\left(-\frac{\textstyle \partial}{\partial x}(a\,u)-\varepsilon\,D^\alpha_t\right),\\
[8pt]
&\approx& \displaystyle
-\frac{\textstyle \partial}{\textstyle \partial x}\left(a\,u+b\,\frac{\textstyle u^2}{\textstyle 2}\right)+\varepsilon\,a\,I_t^{1-\alpha}\left(\frac{\textstyle \partial u}{\partial x}\right)+\varepsilon^2I_t^{1-2\alpha}u.
\end{array}
\end{equation}
Neglecting the $\varepsilon^2$ term and setting $c=\varepsilon\,a$, one obtains
\begin{equation}
\frac{\textstyle \partial u}{\partial t}+\frac{\textstyle \partial}{\textstyle \partial x}\left(a\,u+b\,\frac{\textstyle u^2}{\textstyle 2}\right)=c\,I_t^{1-\alpha}\left(\frac{\textstyle \partial u}{\partial x}\right).
\label{Chester}
\end{equation}
If $\alpha=1/2$, then (\ref{Chester}) recovers Chester's equation \cite{Chester64} modeling the propagation of simple nonlinear waves in a tube with viscothermal losses. In this latter case, the physical parameters are the ratio of specific heats at constant pressure and volume $\gamma$; the pressure at equilibrium $p_0$; the density at equilibrium $\rho_0$; the Prandtl number Pr; the kinematic viscosity $\nu$. It provides physical sense to the coefficients of (\ref{TM}):
\begin{equation}
a=\sqrt{\frac{\textstyle \gamma\,p_0}{\textstyle \rho_0}},\quad b=\frac{\textstyle \gamma+1}{\textstyle 2},\quad c=\left(1+\frac{\textstyle \gamma-1}{\textstyle \sqrt{\mbox{Pr}}}\right)\frac{\textstyle a \sqrt{\nu}}{\textstyle R},
\label{CoeffChester}
\end{equation}
where $a$ is the sound celerity and $R$ is the radius of the tube. The Mach number $M$ and the characteristic angular frequency $\omega_c=2\,\pi\,f_c$ are defined by
\begin{equation}
M=\frac{\textstyle u}{\textstyle a}, \qquad
\omega_c=\frac{\textstyle 2\,\pi\,a}{\textstyle \lambda_c},
\label{Mach}
\end{equation}
where $\lambda_c$ is the wavelength of the wave.


\section{Computation of the diffusive representation}\label{SecProofDR}

We follow the formalism of \cite{Diethelm08} to prove (\ref{DR}). One recalls the definition of the $\Gamma$ function (with $\beta \in \mathbb{R}^{+*}$) and Euler's reflection formula: 
\begin{equation}
\Gamma(\beta)=\int_0^{+\infty}e^{-z}\,z^{\beta-1}\,dz,\hspace{1cm}
\Gamma(1-z)\,\Gamma(z)=\frac{\pi}{\sin \pi z}, \,\forall z\notin \mathbb{Z}.
\label{EulerGamma}
\end{equation}
Based on (\ref{EulerGamma}), the Caputo fractional derivative (\ref{DR}) writes
\begin{equation}
\begin{array}{lll}
D_t^\alpha h 
&=& \ds \frac{1}{\Gamma(1-\alpha)}\int_{0}^{t}(t-\tau)^{-\alpha}\frac{dh}{d\tau}(\tau)\,d\tau,\\
[10pt]
&=& \ds \frac{\sin \pi\alpha}{\pi}\left(\int_0^{+\infty}e^{-z}\,z^{\alpha-1}\,dz\right)\left(\int_{0}^{t}(t-\tau)^{-\alpha}\frac{dh}{d\tau}(\tau)\,d\tau\right),\\
[10pt]
&=& \ds \frac{\sin \pi\alpha}{\pi}\int_0^t\left(\int_0^{+\infty}e^{-z}\left(\frac{z}{t-\tau}\right)^\alpha\frac{1}{z}\,dz\right)\,\frac{dh}{d\tau}(\tau)\,d\tau.
\end{array}
\end{equation}
Setting the change of variables $z=(t-\tau)\,\theta^2$ in the inner integral w.r.t. $z$, we find
\begin{equation}
\begin{array}{lll}
D_t^\alpha h 
&=& \ds \frac{2\,\sin \pi\alpha}{\pi}\int_0^t\left(\int_0^{+\infty}\theta^{2\alpha-1}\,e^{-(t-\tau)\,\theta^2}\,d\theta  \right)\,\frac{dh}{d\tau}(\tau)\,d\tau,\\
[10pt]
&=& \ds \int_0^{+\infty}{\phi(x,t,\theta)\,d\theta},
\end{array}
\end{equation}
by using Fubini's theorem. One recovers (\ref{DR}) with the diffusive variable (\ref{VarDiff}).


\section{Extended diffusive representation}\label{SecProofRD}

Here we prove proposition \ref{PropRD}.

\begin{proof}
\noindent
From (\ref{FI}), the fractional integral of order $\beta$ is
\begin{equation}
\begin{array}{lll}
\ds I_t^\beta u &=& \ds \frac{t^{\beta-1}}{\Gamma(\beta)}\mathop{*}\limits_{t} u,\\
[8pt]
&=& \ds \int_0^{+\infty}\gamma_\beta\,\theta^{1-2\beta}\left(\int_0^t u(x,\tau)\,e^{-(t-\tau)\,\theta^2}d\tau\right)\,d\theta.
\end{array}
\label{ProofRD1}
\end{equation}
Equation (\ref{ProofRD1}) is compared with (\ref{VarDiff})-(\ref{ODE-DR}): replacing $\frac{\partial u}{\partial t}$ by $u$, and $\alpha$ by $\beta=1-\alpha$, gives the desired result (\ref{DR_I}).

The non-homogeneous ordinary differential equation (\ref{ODE_Xi}) is integrated:
\begin{equation}
\psi(x,t,\theta)=\Psi(x,\theta)\,e^{-\theta^2 t}+\gamma_\beta\,\theta^{1-2\beta}\int_0^t u(x,\tau)\ e^{-(t-\tau)\,\theta^2} d\tau.
\label{ProofRD2}
\end{equation}
Comparison between (\ref{ProofRD1}) and (\ref{ProofRD2}) yields
\begin{equation}
I_t^\beta u=\int_0^{+\infty}\left(\psi(x,t,\theta)-\Psi(x,\theta)\,e^{-\theta^2 t}\right)\,d\theta.
\label{ProofRD3}
\end{equation}
Based on (\ref{ProofRD3}), one gets
\begin{equation}
\begin{array}{lll}
\ds D_t^\alpha u &=& \ds \frac{d}{dt}\left(I_t^{1-\alpha} u\right)-\frac{t^{-\alpha}}{\Gamma(1-\alpha)}\,u_0(x),\\
[10pt]
&=& \ds \frac{d}{dt}\int_0^{+\infty}\left(\psi(x,t,\theta)-\Psi(x,\theta)\,e^{-\theta^2 t}\right)\,d\theta-\frac{t^{-\alpha}}{\Gamma(1-\alpha)}\,u_0(x),\\
[10pt]
&=& \ds \underbrace{\int_0^{+\infty}\frac{\partial \psi}{\partial t}(x,t,\theta)\,d\theta}_{\Delta_1}+\underbrace{\int_0^{+\infty}\Psi(x,\theta)\,\theta^2\,e^{-\theta^2 t}\,d\theta}_{\Delta_2}-\frac{t^{-\alpha}}{\Gamma(1-\alpha)}\,u_0(x).
\end{array}
\label{ProofRD4}
\end{equation}
From (\ref{ODE_Xi}), the first term in the r.h.s. of (\ref{ProofRD4}) writes
\begin{equation}
\Delta_1=\int_{0}^{+\infty} \left(-\theta^2\,\psi+\gamma_\beta\,\theta^{1-2\beta}\,u\right)\,d\theta.
\label{ProofRD5}
\end{equation}
Using the initial condition given in (\ref{DiffusifEtendu}), the second term in the r.h.s. of (\ref{ProofRD4}) writes
\begin{equation}
\begin{array}{lll}
\ds \Delta_2 &=& \ds \int_0^{+\infty}\gamma_\beta\,\frac{u_0(x)}{\theta^{1+2\beta}}\,\theta^2\,e^{-\theta^2 t}\,d\theta,\\
[10pt]
&=& \ds u_0(x)\,\gamma_\beta \int_0^{+\infty}\theta^{2\alpha-1}\,e^{-\theta^2 t}\,d\theta,\\
[10pt]
&=& u_0(x)\,\gamma_\alpha\,t^{-\alpha}\,\Gamma(\alpha),\\
[10pt]
&=& \ds u_0(x)\,\frac{t^{-\alpha}}{\Gamma(1-\alpha)},
\end{array}
\label{ProofRD6}
\end{equation}
where we have used $\gamma_\alpha=\gamma_\beta$ and the classical identity $\Gamma(\alpha)\,\Gamma(1-\alpha)=\frac{\pi}{\sin \pi \alpha}$. Injecting (\ref{ProofRD5}) and (\ref{ProofRD6}) into (\ref{ProofRD4}), one recovers (\ref{DiffusifEtendu}), which concludes the proof.
\end{proof} 


\section{Spectrum of the diffusive matrix ${\bf S}$}\label{SecProofS}

Here we prove proposition \ref{PropVpS}.

\begin{proof}
Let $P_{\bf S}(\lambda)$ denote the characteristic polynomial of the matrix ${\bf S}$, i.e. $P_{\bf S}(\lambda)= \det({\bf S}-\lambda\,{\bf I}_{L+1})$ with ${\bf I}_{L+1}$ the $(L+1)$-identity matrix. The line $i$ and the column $j$ of the determinant are denoted by ${\cal L}_i$ and ${\cal C}_j$, respectively. The following algebraic manipulations are performed successively: 
\begin{itemize}
\item[(i)] ${\cal L}_j\leftarrow {\cal L}_j-\gamma_\alpha\,\theta_j^{2\alpha-1}\,{\cal L}_0\text{ with }j=1,\dots,\,L$
\item[(ii)] ${\cal C}_1\leftarrow {\cal C}_1 \prod\limits_{\ell=1}^L(-\theta_\ell^2-\lambda)$
\item[(iii)] ${\cal C}_1\leftarrow {\cal C}_1-\gamma_\alpha\,\theta_\ell^{2\alpha-1}\,\lambda\,{\cal C}_\ell\prod\limits_{\substack{i=1\\i\neq \ell}}^L(-\theta_i^2-\lambda)\mbox{ for }\ell=2,\dots,\, L+1$.
\end{itemize}
It follows
$$
P_{\bf S}(\lambda)\,\prod\limits_{\ell=1}^L(-\theta_\ell^2-\lambda)
=\lambda\,Q_{\bf S}(\lambda)\,\prod\limits_{\ell=1}^L(-\theta_\ell^2-\lambda)
$$
with
$$
Q_{\bf S}(\lambda)=\prod\limits_{\ell=1}^L(-\theta_\ell^2-\lambda)+\varepsilon\,\gamma_\alpha\sum\limits_{\ell=1}^L\mu_\ell\,\theta_\ell^{2\alpha-1}\prod\limits_{\substack{i=1\\i\neq \ell}}^L(-\theta_i^2-\lambda).
$$
Since $Q_{\bf S}(-\theta_\ell^2)\neq 0$, one gets
$$
P_{\bf S}(\lambda)=\lambda\,Q_{\bf S}(\lambda).
$$
The roots of $P_{\bf S}(\lambda)$ are studied in 4 steps.\\

\noindent
\underline{Step 1}. One has $P_{\bf S}(0)=0$ and $Q_{\bf S}(0)\neq0$, therefore 0 is a simple eigenvalue of ${\bf S}$. \\

\noindent
\underline{Step 2}. In the limit $\lambda \rightarrow 0$, one obtains
$$
Q_{\bf S}(\lambda)\mathop{\sim}\limits_{0}(-1)^{L+1}\left(\prod\limits_{\ell=1}^L\theta_\ell^2+\varepsilon\, \gamma_\alpha\sum\limits_{\ell=1}^L\mu_\ell\,\theta_\ell^{2\alpha-1}\prod\limits_{\substack{i=1\\i\neq \ell}}^L\theta_i^2\right),
$$
hence $\mbox{sgn}(P_{\bf S}(0^-))=(-1)^{L+2}=(-1)^L$.\\

\noindent
\underline{Step 3}. At the quadrature nodes, one has ($j=1,\cdots,\,L$)
$$
\begin{array}{lll}
P_{\bf S}(-\theta_j^2)&=& \displaystyle -\theta_j^2\,Q_{\bf S}(-\theta_j^2),\\
[8pt]
&=& \displaystyle-\theta_j^2\,\varepsilon\,\gamma_\alpha\sum\limits_{\ell=1}^L\mu_\ell\,\theta_\ell^{2\alpha-1}\prod\limits_{\substack{i=1\\i\neq \ell}}^L(-\theta_i^2+\theta_j^2),\\
&=& \displaystyle -\varepsilon\,\gamma_\alpha\,\mu_j\,\theta_j^{2\alpha+1}\prod\limits_{\substack{i=1\\i\neq j}}^L(\theta_j^2-\theta_i^2),
\end{array}
$$
and hence $\mbox{sgn}(P_{\bf S}(-\theta_j^2))=(-1)^{L-j+1}$.\\

\noindent
\underline{Step 4}. In the limit $\lambda \rightarrow-\infty$, one has
$$
Q_{\bf S}(\lambda)\mathop{\sim}\limits_{-\infty}(-1)^{L+1}\,\lambda^L\,\Rightarrow\,P_{\bf S}(\lambda)\mathop{\sim}\limits_{-\infty}(-1)^{L+1}\,\lambda^{L+1}=|\lambda|^{L+1},
$$
and hence $\mbox{sgn}(P_{\bf S}(-\infty))=+1$.\\

The sign of the characteristic polynomial is summed up in the following table:
$$
\begin{array}{c|cccccccc}
\lambda & -\infty & -\theta_L^2 & \cdots & -\theta_{\ell+1}^2 & -\theta_\ell^2 & \cdots & -\theta_1^2 & 0\\
[6pt]
\hline
\\
\mbox{sgn}(P_{\bf S}(\lambda)) & +1 & -1 & \cdots & (-1)^{L-\ell} & (-1)^{L-\ell+1} & \cdots & (-1)^L & (-1)^L
\end{array}
$$
We introduce the intervals
$$
I_\ell=\left\{
\begin{array}{l}
]-\theta_{\ell+1}^2,-\theta_\ell^2], \hspace{0.5cm}\ell=1,\cdots,\,L-1,\\
\\
]-\infty,-\theta_L^2],\hspace{0.65cm} \ell=L.
\end{array}
\right.
$$
Given that $P_{\bf S}$ is continuous, the previous table shows that the polynomial $P_{\bf S}$ changes signe in each of the intervals $I_\ell$. Consequently, one deduces that $P_{\bf S}$ vanishes at least once on each $I_\ell$, i.e. $L$ times. Lastly, $P_{\bf S}$ owns at most $L$ nonzero roots. Consequently $\exists \,!\,\lambda_\ell\in I_\ell\,/\,P_{\bf S}(\lambda_\ell)=0$, with $\ell=1,\dots,\,L$. 
\quad \end{proof}\\


{\bf Acknowledgments}. We are grateful to Dr Y. Diouane, from ISAE, for his useful comments on the manuscript. We thank also the anonymous Reviewers for their constructive remarks.


\end{document}